\documentclass[aps,pra,showpacs,twocolumn]{revtex4}
\usepackage{amsmath, amssymb}
\usepackage{graphicx}
\usepackage{subfigure}
\usepackage{xcolor}
\usepackage{amsthm}

\newtheorem{thrm}{Theorem}

\renewcommand{\vec}[1]{\mathbf{#1}}

\DeclareMathOperator{\Li}{Li}

\begin{document}

\title{Multipartite minimum uncertainty products}

\author{E. Shchukin}
\email{evgeny.shchukin@gmail.com}
\affiliation{Lehrstuhl f\"ur Modellierung und Simulation,
Universit\"at Rostock, D-18051 Rostock, Germany}

\begin{abstract}
In our previous work we have found a lower bound for the multipartite uncertainty product of the position and momentum observables over all separable states.
In this work we are trying to minimize this uncertainty product over a broader class of states to find the fundamental limits imposed by nature on the observable quantites.
We show that it is necessary to consider pure states only and find the infimum of the uncertainty product over a special class of pure states 
(states with spherically symmetric wave functions). It is shown
that this infimum is not attained. We also explicitly construct a parametrized family of states that 
approaches the infimum by varying the parameter. Since the constructed states beat the lower bound for separable states, they are entangled.
We thus show that there is a gap that separates the values of a simple measurable quantity for separable states from entangled ones and we also try 
to find the size of this gap.
\end{abstract}


\maketitle

\section{Introduction}

The famous Heisenberg uncertainty relation states that the standard deviation of position and momentum observables of any quantum states
cannot be simultaneously small or, more precisely, that (in appropriate units)
\begin{equation}\label{eq:hei}
    \sigma_x \sigma_p \geqslant \frac{1}{2}.
\end{equation}
This inequality is tight and becomes equality for the vacuum state whose wave function in the position representation is
\begin{equation}\label{eq:vac}
    \psi_0(x) = \frac{1}{\sqrt[4]{\pi}} e^{-x^2/2}.
\end{equation}
In \cite{PhysRevA.84.052325} we have shown that the inequalities \eqref{eq:hei} for different degrees of freedom can be 
"multiplied" side-by-side to produce inequalities valid for all completely separable states. For example, the inequalities
\begin{equation}\label{eq:sxp}
    \sigma_{xx} \sigma_{pp} \geqslant \frac{1}{4} \quad \text{and} \quad
    \sigma_{xp} \sigma_{px} \geqslant \frac{1}{4}
\end{equation}
are valid for all bipartite separable states. These inequalities for separable states are also tight since for the two-dimensional
vacuum they become equality. We have also shown that these inequalities can be violated. A question of fundamental interest is 
what is infimum of the products on left hand side of the inequalities \eqref{eq:sxp} over the set of all
quantum states?

For any physical quantum state these products are strictly positive, but this does not mean that the infimum is also positive.
The problem of determining this infimum is surprisingly more complicated than in a single partite case. In this work
we do not completely solve this problem, we restrict it to the class of states with real spherically symmetric wave functions and find 
the infimum over this class of states. This infimum is not attained, so we construct family of states parametrized
by a real parameter $0 < \xi < 1$ such that when $\xi \to 1$ the uncertainty product approaches the infimum. More precisely, 
we construct a family of states such that $\sigma_{xp} \sigma_{px} < 1/4$ for all $0 < \xi < 1$ and
$\sigma_{xp} \sigma_{px} \to 1/8$
when $\xi \to 1$. We also generalize our construction to a larger number of parties. As the number of parties growth,
the problem becomes more and more complicated, so we obtain explicit
results only for four-partite and six-partite cases and outline the general approach.

The paper is organized as follows. In section \ref{sec:up} we set up the environment in which we work in this paper and prove
that to minimize the uncertainty product under consideration pure states are sufficient. 
In section \ref{sec:simple-state} we derive an analytical expression for the state obtained numerically in our previous work and 
introduce the technique that is used to obtain the main result of this paper. In section \ref{sec:s} we obtain our main result, i.e. 
we derive family of bipartite states
that minimize the uncertainty product in the special class of pure states with real spherically symmetric wave functions. 
The derivation is not difficult but lengthy, so it is divided 
into several steps to make it easier to follow. In section \ref{sec:p} we analyze some properties of these states. In section 
\ref{sec:g} we generalize our construction to a more general case of a larger number of parties. The conclusion is
in section VII. Proofs of the auxiliary results are given in the appendices 
at the end of the paper. We make an intensive use of the book \cite{gr}, which we refer to simply as GR.
It is rather a large book with a lot of results, so we give the page and formula numbers to easily locate the results we refer to.

\section{Multipartite uncertainty product} \label{sec:up}

The general form of the uncertainty relation \eqref{eq:hei} states that for arbitrary observables $\hat{A}$ and $\hat{A}'$ and 
for all quantum states the inequality 
$\sigma_A \sigma_{A'} \geqslant 1/2|\langle[\hat{A}, \hat{A}']\rangle|$ is valid, 
where $\sigma_A = \sqrt{\langle\hat{A}^2\rangle - \langle\hat{A}\rangle^2}$ is the standard deviation of the observable $\hat{A}$. 
It has been shown in \cite{PhysRevA.84.052325} that all completely separable $N$-partite states satisfy the inequality
\begin{equation}
    \sigma_{A_1 \ldots A_N} \sigma_{A'_1 \ldots A'_N} \geqslant \frac{1}{2^N} |\langle[\hat{A}_1, \hat{A}'_1] \ldots [\hat{A}_N, \hat{A}'_N]\rangle|
\end{equation}
for arbitrary observables $\hat{A}_i$ and $\hat{A}'_i$ acting on the $i$th part, $i = 1, \ldots, N$.
If we take $\hat{A}_i = \hat{x}_i$, $\hat{A}'_i=\hat{p}_i$ for all $i$, we then get the inequality
\begin{equation}\label{eq:xp}
    \sigma_{x \ldots x} \sigma_{p \ldots p} \geqslant \frac{1}{2^N}.
\end{equation}
This inequality is tight --- it is easy to verify that for the $N$-partite vacuum state with the wave-function
\begin{equation}
    \psi_0(x_1, \ldots, x_N) = \frac{1}{\sqrt[4]{\pi^N}} e^{-(x^2_1 + \ldots + x^2_N)/2}
\end{equation}
the left-hand side of the inequality \eqref{eq:xp} is exactly $2^{-N}$. The natural question is --- how strong can the inequality \eqref{eq:xp}
be violated?

First of all note that by analogy with the inequality \eqref{eq:xp} the following $2^{N-1}$ inequalities are also valid:
\begin{equation}\label{eq:xpxp}
    \sigma_{x p \ldots} \sigma_{p x \ldots} \geqslant \frac{1}{2^N},
\end{equation}
where the two sequences of $x$ and $p$ complement each other --- if one of them has $x$ in some position then the other has $p$ in the same position.
There are $2^N$ sequences of $x$ and $p$ of length $N$, but the number of different inequalities is only half of this number, i.e. $2^{N-1}$.  
If we can find a state that violates one inequality of this kind then we can find states that violate any such inequality. 
To do it we need the phase-shifting operator \cite{leonhardt}
\begin{equation}
    \hat{U}_\varphi = e^{-i\varphi\hat{n}} = \exp\left(-\frac{i\varphi}{2}\left(x^2 - \frac{d^2}{dx^2}-1\right)\right).
\end{equation}
The action of this operator on wave functions is simply the fractional Fourier transform \cite{ff}
\begin{equation}
\begin{split}
    \mathcal{F}_\varphi[\psi(x)](p) &= \frac{e^{i(\varphi/2-\pi/4)}}{\sqrt{2\pi\sin\varphi}}e^{i\cot\varphi\frac{p^2}{2}} \times \\
    &\int_{\mathbf{R}}\exp\left(i \cot\varphi\frac{x^2}{2} - i \frac{px}{\sin\varphi}\right) \psi(x) \, dx.
\end{split}
\end{equation}
As one can easily see, if $\varphi = \pi/2$ then 
\begin{equation}
    \mathcal{F}_{\pi/2}[\psi(x)](p) = \frac{1}{\sqrt{2\pi}} \int_{\mathbf{R}} \psi(x) e^{-ipx} \, dx \equiv \psi(p)
\end{equation}
is the momentum representation of the state $|\psi\rangle$ with the wave function $\psi(x)$ in the coordinate representation.
Let us take a pure bipartite ($N = 2$) state with the wave function $\psi(x, y)$ and apply partial fractional Fourier transform
$\mathcal{F}_{\pi/2}$ to this wave function with respect to the second argument. Denote this new wave function $\tilde{\psi}$.
Then for the new state we have the relations
\begin{equation}
    \tilde{\sigma}_{xp} = \sigma_{xx}, \quad \tilde{\sigma}_{px} = \sigma_{pp}.
\end{equation}
This means that if the original state violates the inequality \eqref{eq:xp} for $N=2$ then the new state
violates the inequality 
\begin{equation}\label{eq:xppx}
    \sigma_{xp}\sigma_{px} \geqslant 1/4. 
\end{equation}
The similar conclusion is also valid in the general multipartite case --- from a state that violates one of the inequalities \eqref{eq:xpxp}
we can construct $2^{N-1}-1$ states that violate the other $2^{N-1}-1$ inequalities of this form.

If we want to minimize the inequalities \eqref{eq:xpxp} then we can restrict our attention to pure states only due to the
\begin{thrm}
If the inequality 
\begin{equation}\label{eq:AB}
    \sigma_{AB}\sigma_{A'B'} \geqslant \delta > 0
\end{equation}
is valid for all bipartite pure quantum states of some quantum system, where operators $\hat{A}$, $\hat{A}'$ act on one degree of freedom and 
$\hat{B}$, $\hat{B}'$ act on the other one, then this inequality is also valid for all states (i.e. including mixed states). The similar 
statement can be extended to multipartite case.
\end{thrm}
The statement of this theorem is intuitively clear because by mixturing quantum states we can only increase the dispersion.
Equivalently it can be formulated as follows:
\begin{equation}
    \inf_{\text{all}} \sigma_{AB}\sigma_{A'B'} = \inf_{\text{pure}} \sigma_{AB}\sigma_{A'B'},
\end{equation}
where \textit{all} means all bipartite quantum states and \textit{pure} means bipartite pure states only.
So, if we want to minimize the inequalities \eqref{eq:xpxp} then we should focus our efforts on pure states.
It is a very useful result since a wave function is a much simpler object than a density operator (it is not necessary to care about
positivity). 
\begin{proof}
The proof follows the idea of \cite{PhysRevA.84.052325}. We have the inequality
\begin{equation}
    \sigma^2_{AB} + \lambda^2\sigma^2_{A'B'} \geqslant 2 \lambda \sigma_{AB}\sigma_{A'B'} \geqslant 2\lambda \delta,
\end{equation}
which is valid for all bipartite pure states and for all non-negative $\lambda$ by assumption. Now let us take a general mixed state $\varrho$.
It can be written as a mixture of pure states
\begin{equation}
    \varrho = \sum_k p_k |\psi_k\rangle\langle\psi_k|,
\end{equation}
where $p_k$ are non-negative numbers such that $\sum_k p_k = 1$ and $|\psi_k\rangle$ are some bipartite (pure) states. 
Due to the concavity of the dispersion, we have
\begin{equation}
\begin{split}
    \sigma^2_{AB}(\varrho) &+ \lambda^2\sigma^2_{A'B'}(\varrho) \\
    &\geqslant \sum_k p_k(\sigma^2_{AB}(|\psi_k\rangle) + \lambda^2\sigma^2_{A'B'}(|\psi_k\rangle)) \\
    & \geqslant 2\lambda \delta \sum_k p_k = 2\lambda \delta.
\end{split}
\end{equation}
From Lemma 3 of \cite{PhysRevA.84.052325} we conclude that 
\begin{equation}
    4\delta^2 \leqslant 4 \sigma^2_{AB}(\varrho)\sigma^2_{A'B'}(\varrho),
\end{equation}
which is equivalent to the inequality \eqref{eq:AB} for the mixed state $\varrho$.
\end{proof}

As a measure of violation of the inequalities \eqref{eq:xpxp} we take the ratio of the right-hand to the left-hand side. 
In \cite{PhysRevA.84.052325} a state violating the inequality \eqref{eq:xppx} has been constructed. The violation of this inequality
for that state is $\approx 1.2192$. In this work we try to do better and
find the minimal value of the product $\sigma_{xp} \sigma_{px}$ for a special class of states, thus constructing highly entangled (with respect to this 
inequality) states and generalize our construction for larger numbers of parties.

\section{Simple state}\label{sec:simple-state}

We start our discussion of states that violate the inequality \eqref{eq:xppx} with the derivation of an analytically expression for the coefficients of the state
\begin{equation}\label{eq:psi}
    |\psi\rangle = \sum^{+\infty}_{n=0} c_n |2n,2n\rangle
\end{equation}
that has been used in \cite{PhysRevA.84.052325} and get an exact analytical expression for the violation value of $1.2192$
that has been previously obtained numerically. It has been shown that the difference
$\sigma_{xp} \sigma_{px} - 1/4$ for this state is given by the following quadratic form of the coefficients $c_n$:
\begin{equation}\label{eq:Q}
    Q = \sum^{+\infty}_{n=0} \bigl(2n(2n+1)c^2_n -(n+1)(2n+1)c_n c_{n+1}\bigr).
\end{equation}
This quadratic form has been minimized by numerically computing the minimal eigenvalue $\lambda_{\mathrm{min}} \approx -0.04495$ of the 
truncated matrix corresponding to this form.  The components
of the normalized eigenvector corresponding to the minimal eigenvalue are the coefficients of the state \eqref{eq:psi} and for this state we have
$\sigma_{xp} \sigma_{px} = (1/4) + \lambda_{\mathrm{min}} \approx 0.20505$. 
We now derive an analytical expression for these coefficients. The method presented here for this simpler case will be useful later when we present technique to solve
more complicated and more general problems.

Factoring out the common terms under the summation sign in the form \eqref{eq:Q}, we arrive to a simpler expression
\begin{equation}
    Q = -c_0 c_1 + \sum^{+\infty}_{n=1}\left(2+\frac{1}{n}\right)\tilde{c}_n(2\tilde{c}_n - \tilde{c}_{n+1}),
\end{equation}
where we use the coefficient transformation $\tilde{c}_n = n c_n$, $n \geqslant 1$. We do not transform the 
coefficient $c_0$ and the coefficient $c_1$ is left unchanged: $\tilde{c}_1 = c_1$. 
Let us assume that the transformed coefficients satisfy the simple relation $\tilde{c}_{n+1} = \xi \tilde{c}_n$, 
$n \geqslant 1$, where $|\xi| < 1$ is a real parameter to be determined, i.e. that they form a geometrical progression. 
From this relation we can easily derive that $\tilde{c}_n = \xi^{n-1}c_1$, $n \geqslant 1$ and then express 
the quadratic form $Q$ as a function of $c_0$ and $c_1$ only
\begin{equation}\label{eq:Q2}
\begin{split}
    Q &= -c_0 c_1 + c^2_1 (2-\xi) \sum^{+\infty}_{n=1}\left(2+\frac{1}{n}\right)\xi^{2n-2} \\
      &= -c_0 c_1 + (2-\xi)\left(\frac{2}{1-\xi^2} - \frac{\ln(1-\xi^2)}{\xi^2}\right) c^2_1.
\end{split}
\end{equation}
On the other hand, from the relation for the transformed coefficients we immediately obtain the expression for the original coefficients $c_n$
\begin{equation}\label{eq:cn}
    c_n = \frac{\xi^{n-1}}{n}c_1, \quad n \geqslant 1.
\end{equation}
The first two coefficients $c_0$ and $c_1$ are not independent, from the normalization condition of the state \eqref{eq:psi} we obtain the following relation between them:
\begin{equation}\label{eq:c0c1}
    c^2_0 + c^2_1 \sum^{+\infty}_{n=1} \frac{\xi^{2n-2}}{n^2} = c^2_0 + \frac{\Li_2(\xi^2)}{\xi^2} c^2_1 = 1,
\end{equation}
where $\Li_s(z)$ is the polylogarithm special function defined by the infinite power series as \cite{prudnikov}
\begin{equation}
    \Li_s(z) = \sum^{+\infty}_{n=1} \frac{z^n}{n^s}.
\end{equation}
The equation \eqref{eq:c0c1} describes an ellipse and this ellipse can be parametrized as follows:
\begin{equation}\label{eq:c01}
    c_0 = \cos \varphi, \quad c_1 = \frac{\xi}{\sqrt{\Li_2(\xi^2)}} \sin\varphi.
\end{equation}
Substituting these expressions into the equation \eqref{eq:Q2}, we obtain the following expression for our quadratic form:
\begin{equation}\label{eq:Q3}
    Q = -C_1(\xi) \sin(2\varphi)+C_2(\xi)(1-\cos(2\varphi)),
\end{equation}
where the coefficients $C_1(\xi)$ and $C_2(\xi)$ are given by
\begin{equation}
\begin{split}
    C_1(\xi) &= \frac{1}{2}\frac{\xi}{\sqrt{\Li_2(\xi^2)}}, \\
    C_2(\xi) &= \frac{(2-\xi)}{2}\left(\frac{2}{1-\xi^2} - \frac{\ln(1-\xi^2)}{\xi^2}\right)\frac{\xi^2}{\Li_2(\xi^2)}.
\end{split}
\end{equation}
The minimal value of the expression \eqref{eq:Q3} when $\xi$ is fixed and $\varphi$ varies is
\begin{equation}\label{eq:Q0}
    Q_0(\xi) = C_2(\xi) - \sqrt{C^2_1(\xi)+C^2_2(\xi)},
\end{equation}
and the value of the angle $\varphi$ for which this minimal value is attained is determined from the equation
\begin{equation}\label{eq:phi}
    \tan(2\varphi) = \frac{C_1(\xi)}{C_2(\xi)}.
\end{equation}
The plot of the function \eqref{eq:Q0} is shown in Fig.~\ref{fig:psi1}. One can numerically minimize this function (for example, with \textsl{Mathematica}) and get that 
the minimal value is $\approx -0.04495$ for $\xi_{\mathrm{min}} \approx 0.318674$. This result is in perfect agreement with the previous result obtained 
numerically by computing the eigenvalues of the quadratic form \eqref{eq:Q}. To obtain the coefficients $c_n$, we take the parameter $\xi_{\mathrm{min}}$ 
and find the corresponding angle $\varphi_{\mathrm{min}}$ from the equation \eqref{eq:phi}. Then, we compute the coefficients $c_0$ and $c_1$ according to the equation
\eqref{eq:c01}. The rest of the coefficients are obtained from the equation \eqref{eq:cn}. Though we have not strictly proved that $-0.04495$ is the minimal eigenvalue
of the quadratic form \eqref{eq:Q}, we have analytically constructed a state on which this value is attained and shown that this result agrees with the numerical computations.

\begin{figure}
\begin{center}
    \includegraphics[scale=0.8]{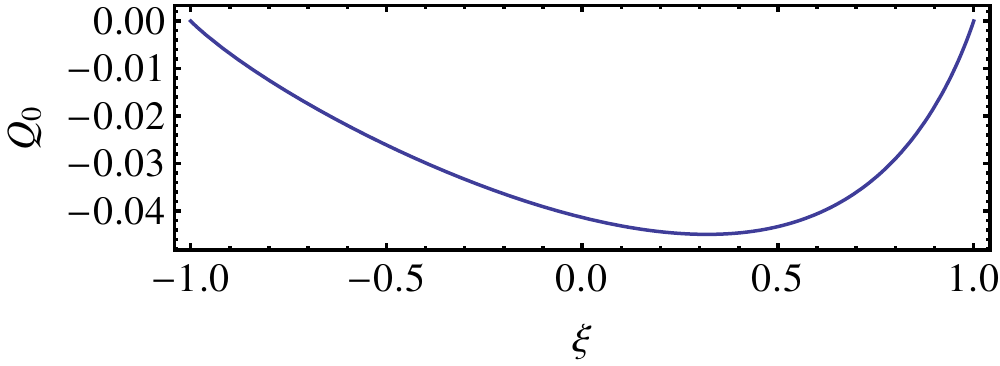}
\end{center}
\caption{The function $Q_0(\xi)$, defined by the equation \eqref{eq:Q0}.}\label{fig:psi1}
\end{figure}

\section{Construction of spherically symmetric states} \label{sec:s}

The state \eqref{eq:psi} has a very special form, so one can ask whether it possible to stronger violate the inequality \eqref{eq:xppx} with a more 
general state. Here, we construct such a family of states that violate this inequality. The construction is lengthy and thus divided into several steps.

\textit{First step.} At this step we analyze the product $\sigma_{xp} \sigma_{px}$ and transform it
into a more manageable form. The variances $\sigma^2_{xp}$ and $\sigma^2_{px}$ read as
\begin{equation}
    \sigma^2_{xp} = \langle\hat{x}^2_a \hat{p}^2_b\rangle - \langle\hat{x}_a \hat{p}_b\rangle^2, \quad
    \sigma^2_{px} = \langle\hat{p}^2_a \hat{x}^2_b\rangle - \langle\hat{p}_a \hat{x}_b\rangle^2.
\end{equation}
Their product would be easier to deal with if the averages $\langle\hat{x}_a \hat{p}_b\rangle$ and $\langle\hat{p}_a \hat{x}_b\rangle$ were zero. 
This is the case when, for example, the wave function $\psi(x, y)$ of the state under study is real. In fact, for a real wave function we have
\begin{equation}
    \langle\hat{x}_a \hat{p}_b\rangle = -i\iint_{\mathbf{R}^2} \psi(x,y) x \frac{\partial \psi}{\partial y}(x,y) \, dx \, dy  = 0,
\end{equation}
since for any fixed $x$ the integral over $y$ can be taken explicitly 
\begin{equation}
    \int \psi \frac{\partial\psi}{\partial y} \, dy = \frac{1}{2}(\psi^2(x,+\infty)-\psi^2(x,-\infty)) = 0,
\end{equation}
and is equal to zero due to the normalization property of $\psi$. In the same way we obtain the equality $\langle\hat{p}_a \hat{x}_b\rangle = 0$.
We have that for a real wave function we can write
\begin{equation}\label{eq:s2}
\begin{split}
    \sigma^2_{xp} \sigma^2_{px} &= \langle\hat{x}^2_a \hat{p}^2_b\rangle \langle\hat{p}^2_a \hat{x}^2_b\rangle \\
    & = \frac{1}{4}\bigl(\langle\hat{x}^2_a\hat{p}^2_b + \hat{p}^2_a\hat{x}^2_b\rangle^2 - \langle\hat{x}^2_a\hat{p}^2_b - \hat{p}^2_a\hat{x}^2_b\rangle^2\bigr)
\end{split}
\end{equation}
and conclude that the product $\sigma_{xp} \sigma_{px}$ can be bounded from above by the square root of the first term on the right-hand side of this expression,
i.e. the following inequality is valid:
\begin{equation}
    \sigma_{xp} \sigma_{px} \leqslant \frac{1}{2}\langle\hat{x}^2_a\hat{p}^2_b + \hat{p}^2_a\hat{x}^2_b\rangle \equiv \langle\hat{Z}\rangle.
\end{equation}
So, at this step, we have reduced our minimization problem to the study of the operator $\hat{Z}$ on the set of real wave functions.

\textit{Second step.} In the standard position representation, the operator $\hat{Z}$ is given by the following differential operator:
\begin{equation}
    \hat{Z} = \frac{1}{2}(\hat{x}^2_a\hat{p}^2_b + \hat{p}^2_a\hat{x}^2_b) = -\frac{1}{2}\left(x^2\frac{\partial^2}{\partial y^2} + y^2\frac{\partial^2}{\partial x^2}\right).
\end{equation}
The average value of this operator on a real wave function $\psi(x,y)$ can be computed as
\begin{equation}\label{eq:Zpsi}
    \langle\hat{Z}\rangle = -\frac{1}{2} \iint_{\mathbf{R}^2} 
    \psi \left(x^2\frac{\partial^2 \psi}{\partial y^2} + y^2\frac{\partial^2 \psi}{\partial x^2}\right)\,dx\,dy.
\end{equation}
Integrating by parts, this expression can be transformed as follows:
\begin{equation}
    \int_{\mathbf{R}} \psi \frac{\partial^2 \psi}{\partial x^2} \, dx = \left.\psi \frac{\partial \psi}{\partial x} \right|^{+\infty}_{-\infty} - 
    \int_{\mathbf{R}} \left(\frac{\partial \psi}{\partial x}\right)^2 \, dx.
\end{equation}
For a normalizable wave function the first term on the right-hand side is zero and, transforming the other term under the integral in Eq.~\eqref{eq:Zpsi}, 
we can conclude that for the average $\langle\hat{Z}\rangle$ we have
\begin{equation}\label{eq:Z}
    \langle\hat{Z}\rangle = \frac{1}{2} \iint_{\mathbf{R}^2} \left(x^2 \left(\frac{\partial \psi}{\partial y}\right)^2 + 
    y^2 \left(\frac{\partial \psi}{\partial x}\right)^2\right)  \, dx \, dy.
\end{equation}
Looking at this expression, it seems natural to write it in polar coordinates. The derivatives with respect to the Cartesian coordinates $x$ and $y$ can be expressed
in terms of the derivatives with respect to polar coordinates $r$ and $\theta$ as follows:
\begin{equation}\label{eq:xy}
\begin{split}
    \frac{\partial}{\partial x} &= \cos\theta \frac{\partial}{\partial r} - \sin\theta \frac{1}{r}\frac{\partial}{\partial \theta}, \\
    \frac{\partial}{\partial y} &= \sin\theta \frac{\partial}{\partial r} + \cos\theta \frac{1}{r}\frac{\partial}{\partial \theta}.
\end{split}
\end{equation}
If we substitute these derivatives into the integral on the right-hand side of Eq.~\eqref{eq:Z}, the resulting expression will not look
any simpler, so we will consider only real spherically symmetric wave function $\psi(x, y)$, i.e. functions that do not depend on the angle $\theta$ when 
written in polar coordinates as $\Psi(r, \theta) = \psi(r \cos\theta, r \sin\theta)$. 
Then $\Psi(r, \theta) \equiv \Psi(r)$, so $\partial \Psi/\partial \theta = 0$ and Eq.~\eqref{eq:Z} takes a simpler form
\begin{equation}\label{eq:Z2}
\begin{split}
    \langle\hat{Z}\rangle &= \frac{1}{4} \int^{+\infty}_0 \int^{2\pi}_0 r^2 \sin^2(2\theta) \left(\frac{d \Psi}{d r}\right)^2 r \,dr \, d\theta \\
    &= \frac{\pi}{4} \int^{+\infty}_0 r^3 \left(\frac{d \Psi}{d r}\right)^2 \, dr.
\end{split}
\end{equation}
The normalization of the real wave function $\psi(x, y)$ can be written as follows:
\begin{equation}\label{eq:psip}
\begin{split}
    \iint_{\mathbf{R}^2} \psi^2(x, y) \, dx \, dy &= \int^{+\infty}_0 \int^{2\pi}_0 \Psi^2(r)r \, dr \, d\theta \\
      &= 2\pi \int^{+\infty}_0 \Psi^2(r) r \, dr = 1.
\end{split}
\end{equation}
Here we note that the second term on right-hand side of the equation \eqref{eq:s2} is
\begin{equation}
    \hat{x}^2_a\hat{p}^2_b - \hat{p}^2_a\hat{x}^2_b = -x^2 \frac{\partial^2}{\partial y^2} + y^2 \frac{\partial^2}{\partial x^2},
\end{equation}
and for our spherically symmetric wave function $\psi(x, y)$ we have
\begin{equation}
\begin{split}
    &\iint_{\mathbf{R}^2}\psi\left(-x^2 \frac{\partial^2 \psi}{\partial y^2} + y^2 \frac{\partial^2 \psi}{\partial x^2}\right) \, dx \, dy \\
    &= \iint_{\mathbf{R}^2}\psi\left(x^2 \left(\frac{\partial \psi}{\partial y}\right)^2 - y^2 \left(\frac{\partial \psi}{\partial x}\right)^2\right) \, dx \, dy = 0,
\end{split}
\end{equation}
which can be easily obtained with the help of Eqs.~\eqref{eq:xy}. This means that in the case of a real spherically symmetric wave function we simply have 
\begin{equation}\label{eq:sZ}
    \sigma_{xp} \sigma_{px} = \langle\hat{Z}\rangle.
\end{equation}
In this step, our original problem has been further reduced to the problem of minimizing the integral \eqref{eq:Z2} provided that 
the function $\Psi(r)$ satisfies the normalization condition \eqref{eq:psip}.

\textit{Third step.} To simplify the integral \eqref{eq:Z2} and the normalization condition \eqref{eq:psip} let us introduce the function $f(r)$ via the relation
\begin{equation}\label{eq:psip2}
    \Psi(r) = \frac{1}{\sqrt{\pi}} f(r^2).
\end{equation}
The normalization condition \eqref{eq:psip} in terms of this function reads as
\begin{equation}
\begin{split}
    2\pi &\int^{+\infty}_0 \Psi^2(r) r \, dr = 2\pi \int^{+\infty}_0 \frac{1}{\pi} f^2(r^2) r \, dr \\
    &= \int^{+\infty}_0 f^2(r^2) \, dr^2 = \int^{+\infty}_0 f^2(r) \, dr = 1.
\end{split}
\end{equation}
We see that the function $f(r)$ is normalized in the ordinary sense. Now let us write the last integral of Eq.~\eqref{eq:Z2} in terms
of the function $f(r)$. For the derivative we have
\begin{equation}\label{eq:psipd}
    \frac{d \Psi}{dr} = \frac{2}{\sqrt{\pi}} r f'(r^2),
\end{equation}
and, substituting this expression into Eq.~\eqref{eq:Z2}, we get
\begin{equation}\label{eq:Z3}
\begin{split}
    \langle\hat{Z}\rangle &= \frac{\pi}{4} \int^{+\infty}_0 r^3 \left(\frac{d \Psi}{d r}\right)^2 \, dr = \int^{+\infty}_0 r^5 f^{\prime 2}(r^2) \, dr \\
    &= \frac{1}{2}\int^{+\infty}_0 r^4 f^{\prime 2}(r^2) \, dr^2 = \frac{1}{2}\int^{+\infty}_0 r^2 f^{\prime 2}(r) \, dr.
\end{split}
\end{equation}
In this step, we have formulated our problem as follows: find the minimum of the integral
\begin{equation}\label{eq:rf}
    \langle\hat{Z}\rangle = \frac{1}{2}\int^{+\infty}_0 r^2 f^{\prime 2}(r) \, dr
\end{equation}
provided that the function $f(r)$ is normalized
\begin{equation}\label{eq:f}
    \int^{+\infty}_0 f^2(r) \, dr = 1.
\end{equation}
In other words, we need to minimize the functional \eqref{eq:rf} over the set of normalized functions. 
Intuitively it seems to be clear that the minimum is not attained because function for which the integral \eqref{eq:rf}
is small should have the shape of a peak located near the origin $r=0$. Outside the peak the function should be small so that it is nearly constant and thus
the derivative is almost zero, and the more narrow the peak, the smaller the region where $f'(r)$ is nonzero. On the other hand, the function $f$ must 
be normalized, and the more narrow the peak the higher it must be, and then one can expect the the derivative inside the peak is large.
The range of $r$ where $f'(r)$ is nonzero is small but the value of the derivative is large, so that the function must have a special shape for the combination
of these two competing features to give the smallest value of the integral \eqref{eq:rf}.
So, we expect to find a parametrized family of functions which represent infinitesimally narrow and infinitely high peaks of some special form when the parameter varies.

\textit{Fourth step.} Let us show that the value of the right-hand side of Eq.~\eqref{eq:rf} cannot be smaller than $1/8$. 
To prove this, we first compute the integral
\begin{equation}
\begin{split}
    \int^{+\infty}_0 &r f(r) f'(r) \, dr = \frac{1}{2} \left. r f(r)^2 \right|^{+\infty}_0 \\
    &- \frac{1}{2}\int^{+\infty}_0 f^2(r) \, dr = -\frac{1}{2}.
\end{split}
\end{equation}
If we apply the Cauchy-Schwarz inequality 
\begin{equation}\label{eq:cs}
    \left(\int^b_a f(r) g(r) \, dr\right)^2 \leqslant \int^b_a f^2(r) \, dr \int ^b_a g^2(r) \, dr,
\end{equation}
which is valid for all real integrable functions $f(r)$ and $g(r)$ on an interval $[a, b]$, $-\infty \leqslant a < b \leqslant +\infty$, 
to the functions $f(r)$ and $g(r) = r f'(r)$ on the interval $[0, +\infty)$, we get
\begin{equation}
    \frac{1}{4} \leqslant \int^{+\infty}_0 f^2(r) \, dr \int^{+\infty}_0 r^2 f^{\prime 2}(r) \, dr,
\end{equation}
and the normalization condition \eqref{eq:f} gives us the desired inequality
\begin{equation}\label{eq:rf2}
    \langle\hat{Z}\rangle = \frac{1}{2}\int^{+\infty}_0 r^2 f^{\prime 2}(r) \, dr \geqslant \frac{1}{8}.
\end{equation}
Note that, in fact, we have just proved the inequality
\begin{equation}\label{eq:rff}
    \|r f'\|^2 \geqslant \frac{1}{4} \|f\|^2,
\end{equation}
where we use the standard notation for the norm and scalar product of two square-integrable functions $f, g \in L(0, +\infty)$
\begin{equation}
\begin{split}
    (f, g) &= \int^{+\infty}_0 f(r) g(r) \, dr, \\
    \|f\|^2 &= (f, f) = \int^{+\infty}_0 f^2(r) \, dr.
\end{split}
\end{equation}
This notation and the inequality \eqref{eq:rff} will be used for the multipartite generalizations of the results obtained here.

The Cauchy-Schwarz inequality \eqref{eq:cs} becomes equality if and only if the functions $f(r)$ and $g(r)$ are linearly dependent
on the interval $[a, b]$. In our case this condition reads as $r f'(r) + \lambda f(r) = 0$ for all $r \geqslant 0$. 
The general solution of this equation is $f(r) = C r^{-\lambda}$.
One can easily see that such a function is not normalizable on the interval $[0, +\infty)$. This means
that there is no normalized functions $f(r)$ that minimize the inequality \eqref{eq:rf2}. The problem now is to construct such functions $f(r)$
that approach to the lower bound $1/8$ as close as possible. As a hint in this construction we look for functions for which there is a linear 
combination of $r f'(r)$ and $f(r)$ that is small in some sense. 

\textit{Fifth step.} Using the Laguerre functions
\begin{equation}
    l_n(r) = L_n(r) e^{-r/2}, \quad n = 0, 1, \ldots,
\end{equation}
which are complete and orthonormal on the interval $[0, +\infty)$ \cite{szego}, where $L_n(r)$ are Laguerre polynomials \cite{gr-L}
\begin{equation}
    L_n(r) = \frac{e^r}{n!} \frac{d^n}{dr^n} (r^n e^{-r}),
\end{equation}
we can write any function $f(r)$, normalizable on $[0, +\infty)$, as a linear combination of the Laguerre functions
\begin{equation}\label{eq:f2}
    f(r) = \sum^{+\infty}_{n=0} d_n l_n(r).
\end{equation}
The normalization of $f(r)$ and orthonormality of $\{l_n(r)\}$ give us the condition 
\begin{equation}\label{eq:d}
    \sum^{+\infty}_{n=0} d^2_n = 1.
\end{equation}
From the relations for Laguerre polynomials \cite{gr-laguerre} one can derive the following equality:
\begin{equation}
    r l'_n(r) = -\frac{1}{2}(n l_{n-1}(r) + l_n(r) - (n+1)l_{n+1}(r)), 
\end{equation}
which is valid for $n \geqslant 0$ (for $n=0$ assume that $l_{-1}(r) = 0$). From this equality, for a general function \eqref{eq:f2} one can easily obtain the relation
\begin{equation}\label{eq:rfp}
    r f'(r) = -\frac{1}{2} \sum^{+\infty}_{n=0} (-n d_{n-1} + d_n + (n+1)d_{n+1}) l_n(r).
\end{equation}
Using the orthonormality of $\{l_n(r)\}$ one can get
\begin{equation}
\begin{split}
    \int^{+\infty}_0 r^2 f^{\prime 2}(r) \, dr &= \frac{1}{4}\sum^{+\infty}_{n=0} (d_n + (n+1)d_{n+1}-n d_{n-1})^2 \\
    &= \frac{1}{2} + \frac{1}{2}R(d_0, d_1, d_2, \ldots),
\end{split}
\end{equation}
where $R = R(d_0, d_1, d_2, \ldots)$ is the quadratic form 
\begin{equation}\label{eq:R}
    R = \sum^{+\infty}_{n=0}\bigl(n(n+1)d^2_n - (n+1)(n+2)d_n d_{n+2}\bigr).
\end{equation}
From the inequality \eqref{eq:rf2} we can conclude that $R \geqslant -1/2$ provided that the equality \eqref{eq:d} is satisfied.
Going back to Eqs.~\eqref{eq:sZ} and \eqref{eq:Z3}, we have
\begin{equation}\label{eq:sigmaR}
    \sigma_{xp} \sigma_{px} = \langle\hat{Z}\rangle = \frac{1}{2}\int^{+\infty}_0 r^2 f^{\prime 2}(r) \, dr = \frac{1}{4} + \frac{1}{4}R.
\end{equation}
We thus reduced our minimization problem to the minimization problem of the quadratic form \eqref{eq:R}.

\textit{Sixth step.} In this step we construct a parametrized family of states for which $R$ tends to $-1/2$, and thus $\sigma_{xp} \sigma_{px}$ tends to $1/8$, but this minimum
is not attained. Computing the eigenvectors corresponding to the minimal eigenvalue of the truncated matrices of the quadratic $R$ one
can observe that the components with odd indices are zero, so let us set $d_{2n+1} = 0$, $d_{2n} = c_n$, 
$n \geqslant 0$. Then the form \eqref{eq:R} becomes
\begin{equation}\label{eq:R2}
    R = \sum^{+\infty}_{n=0}\bigl(2n(2n+1)c^2_n - 2(n+1)(2n+1)c_n c_{n+1}\bigr).
\end{equation}
This is very similar to the quadratic form $Q$ defined by Eq.~\eqref{eq:Q}, except for the
additional factor 2 in the second term under the summation sign. One can try the same approach we used for
the form $Q$, but in this case it does not give us the optimal state, so another idea is needed.

From Eq.~\eqref{eq:rfp} we can obtain
\begin{equation}\label{eq:rfp2}
\begin{split}
    r f'(r) &+ \frac{1}{2}f(r) = \frac{1}{2} \sum^{+\infty}_{n=0} (n d_{n-1} - (n+1)d_{n+1}) l_n(r) \\
    &= \frac{1}{2} \sum^{+\infty}_{n=0} ((2n+1) c_n - (2n+2)c_{n+1}) l_{2n+1}(r).
\end{split}
\end{equation}
If the relation $c_{n+1}/c_n = (2n+1)/(2n+2)$ were valid for all $n \geqslant 0$, then we would have
linear dependence of $r f'(r)$ and $f(r)$ but, as we know, it is impossible for a normalized function $f(r)$.
And, in fact, one can easily check that the sequence $c_n$, defined by this ratio, cannot be normalized. But we can try to define the coefficients
via
\begin{equation}\label{eq:cn2}
    c_{n+1} = \frac{2n+1}{2n+2}\xi c_n,
\end{equation}
where $0 < \xi < 1$. From this relation one can derive an explicit expression for the coefficient $c_n$
\begin{equation}\label{eq:cn3}
    c_n = \binom{2n}{n} \left(\frac{\xi}{4}\right)^n c_0 = \frac{(2n-1)!!}{(2n)!!}\xi^n c_0
\end{equation}
for $n \geqslant 1$. The normalization reads as
\begin{equation}
    \left(1+\sum^{+\infty}_{n=1}\left(\frac{(2n-1)!!}{(2n)!!}\right)^2\xi^{2n}\right)c^2_0 = 
    \frac{2}{\pi}K(\xi)c^2_0 = 1,
\end{equation}
where $K(\xi)$ is the complete elliptic integral of the first kind \cite{gr-K}
\begin{equation}\label{eq:Kxi}
\begin{split}
    K(\xi) &= \int^{\pi/2}_0 \frac{d\theta}{\sqrt{1-\xi^2\sin^2\theta}} = \frac{\pi}{2}\sum^{+\infty}_{n=0} \binom{2n}{n}^2\left(\frac{\xi^2}{16}\right)^n\\
    &= \frac{\pi}{2}\left(1+\sum^{+\infty}_{n=1}\left(\frac{(2n-1)!!}{(2n)!!}\right)^2\xi^{2n}\right).
\end{split}
\end{equation}
We thus obtain the coefficient $c_0$
\begin{equation}\label{eq:c0}
    c_0 = \sqrt{\frac{\pi}{2K(\xi)}}.
\end{equation}
From Eq.~\eqref{eq:rfp2} we get
\begin{equation}
    r f'(r) + \frac{1}{2} f(r) = \frac{1-\xi}{2}\sum^{+\infty}_{n=0}(2n+1)c_n l_{2n+1}(r),
\end{equation}
and using orthonormality of $\{l_n(r)\}$ again we derive
\begin{equation}\label{eq:rfp3}
\begin{split}
    \int^{+\infty}_0 &\left(r f'(r) + \frac{1}{2} f(r)\right)^2 \, dr = \frac{(1-\xi)^2}{4}
    \sum^{+\infty}_{n=0}(2n+1)^2 c^2_n \\
    &= \frac{2E(\xi)-(1-\xi^2)K(\xi)}{4(1+\xi^2)K(\xi)},
\end{split}
\end{equation}
where $E(\xi)$ is the complete elliptic integral of the second kind \cite{gr-E}
\begin{equation}
\begin{split}
    E(\xi) &= \int^{\pi/2}_0\sqrt{1-\xi^2 \sin^2\theta} \, d\theta \\
    &= \frac{\pi}{2}\left(1-\sum^{+\infty}_{n=1}\left(\frac{(2n-1)!!}{(2n)!!}\right)^2\frac{\xi^{2n}}{2n-1}\right).
\end{split}
\end{equation}
Since $\lim_{\xi \to 1} E(\xi) = 1$ and $\lim_{\xi \to 1} K(\xi) = +\infty$, we can conclude that the integral 
on the left-hand side of the equality \eqref{eq:rfp3} tends to zero when $\xi \to 1$. One can say that the
linear combination $r f'(r) + (1/2) f(r)$ becomes smaller and tends to zero when $\xi \to 1$. This means that our choice of coefficients \eqref{eq:cn2} 
is a good candidate for the minimizing function. 

Now we can write the function $f(r)$ defined by Eq.~\eqref{eq:f2} as (remember that $d_{2n+1} = 0$ and $d_{2n} = c_n$)
\begin{equation}\label{eq:ff}
\begin{split}
    f(&r) = \sum^{+\infty}_{n=0} c_n L_{2n}(r) e^{-r/2} \\
    &= c_0 \left(1+\sum^{+\infty}_{n=1}\frac{\xi^n}{2^{2n-1}}\binom{2n-1}{n}L_{2n}(r)\right)e^{-r/2},
\end{split}
\end{equation}
where $c_0$ is given by Eq.~\eqref{eq:c0}. To simplify the sum in this expression, note that the Laguerre polynomial $L_n(r)$
can be represented as follows \cite{gr-laguerre2}:
\begin{equation}\label{eq:Lint}
    L_n(r) = e^r \int^{+\infty}_0 \frac{t^n}{n!}J_0(2\sqrt{rt}) e^{-t} \, dt,
\end{equation}
where $J_0(z)$ is the Bessel function of the first kind \cite{gr-J}
\begin{equation}
    J_0(z) = \sum^{+\infty}_{n=0} (-1)^n \frac{z^{2n}}{2^{2n}(n!)^2}.
\end{equation}
Substituting the integral representation \eqref{eq:Lint} into Eq.~\eqref{eq:ff} we get
\begin{equation}
    f(r) = c_0 e^{r/2} \int^{+\infty}_0 S(\xi, t) J_0(2\sqrt{rt})e^{-t}\, dt,
\end{equation}
where $S(\xi, t)$ is the sum
\begin{equation}
    S(\xi, t) = 1 + \sum^{+\infty}_{n=1} \frac{\xi^n}{2^{2n-1}}\binom{2n-1}{2n}\frac{t^{2n}}{(2n)!}.
\end{equation}
Note that we can simplify this sum as follows:
\begin{equation}\label{eq:S}
\begin{split}
    S(\xi, t) &= 1 + \sum^{+\infty}_{n=1} \frac{1}{2^{2n-1}}\frac{(2n-1)!}{n!(n-1)!}\frac{(t\sqrt{\xi})^{2n}}{(2n)!} \\
    &= \sum^{+\infty}_{n=0} \frac{(t\sqrt{\xi})^{2n}}{2^{2n}(n!)^2} = I_0(t\sqrt{\xi}),
\end{split}
\end{equation}
where $I_0(z) = J_0(iz)$ is the modified Bessel function of the first kind. Finally, we arrive at the following expression:
\begin{equation}\label{eq:fr}
    f(r) = \sqrt{\frac{\pi}{2K(\xi)}}e^{r/2} \int^{+\infty}_0 I_0(t\sqrt{\xi})J_0(2\sqrt{rt})e^{-t}\,dt.
\end{equation}
In Appendix X we prove that the exchange summation and integration in Eq.~\eqref{eq:ff} is legal. In this step, we have constructed a family of the functions 
\eqref{eq:fr}, parametrized by the real parameter $0 < \xi < 1$, that are good candidates for the functions that approach 
the equality in Eq.~\eqref{eq:rf2} as $\xi \to 1$.

\textit{Seventh step.} Now let us compute the quadratic form \eqref{eq:R2}. Substituting the coefficients \eqref{eq:cn3} into the expression \eqref{eq:R2}, we get
$R = (R_1 - R_2) c^2_0$, where $R_1$ and $R_2$ are defined via
\begin{equation}
\begin{split}
    R_1 &= \sum^{+\infty}_{n=1}\frac{(2n+1)!!}{(2n)!!}\frac{(2n-1)!!}{(2n-2)!!}\xi^{2n}, \\
    R_2 &= \xi \sum^{+\infty}_{n=0}\left(\frac{(2n+1)!!}{(2n)!!}\right)^2\xi^{2n}.
\end{split}
\end{equation}
The first sum can be simplified as follows:
\begin{equation}
\begin{split}
    R_1 &= \sum^{+\infty}_{n=0} 2n(2n+1) \left(\frac{(2n-1)!!}{(2n)!!}\right)^2\xi^{2n} \\
        &= \frac{2}{\pi}\xi \frac{d}{d\xi}\left(K(\xi) + \xi \frac{dK(\xi)}{d\xi}\right) \\
        &= \frac{2}{\pi}\frac{(1+\xi^2) E(\xi) -(1-\xi^2)K(\xi)}{(1-\xi^2)^2}.
\end{split}
\end{equation}
The second sum can be computed analogously
\begin{equation}
\begin{split}
    R_2 &= \xi \sum^{+\infty}_{n=0}(2n+1)^2\left(\frac{(2n-1)!!}{(2n)!!}\right)^2\xi^{2n} \\
        &= \frac{2}{\pi}\xi \left(\xi \frac{d}{d\xi}+1\right)\left(K(\xi)+\xi\frac{dK(\xi)}{d\xi}\right) \\
        &= \frac{2}{\pi}\frac{\xi (2E(\xi)-(1-\xi^2)K(\xi))}{(1-\xi^2)^2}.
\end{split}
\end{equation}
We finally obtain
\begin{equation}\label{eq:R3}
    R = -\frac{1}{1+\xi} + \frac{1}{(1+\xi)^2}\frac{E(\xi)}{K(\xi)}.
\end{equation}
As Fig.~\ref{fig:R} illustrates, $R < 0$ for all $0 < \xi < 1$. From the expression \eqref{eq:R3} we also see 
that $R \to -1/2$ when $\xi \to 1$. In this step we have proved that with the choice of the coefficients \eqref{eq:cn2}
the form \eqref{eq:R2} tends to its minimal value $-1/2$, but does not attain it.

\begin{figure}
\includegraphics[scale=0.8]{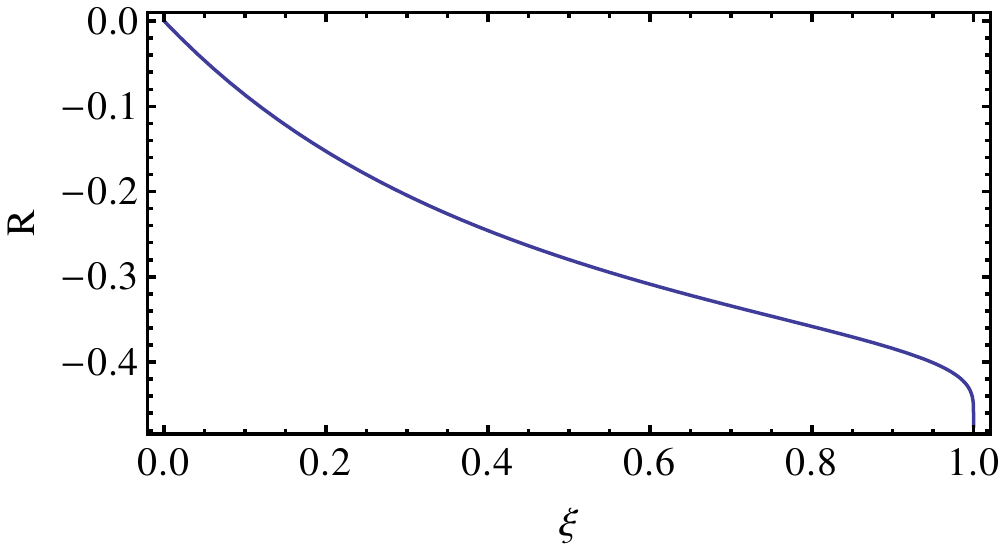}
\caption{The quantity $R$ defined by Eq.~\eqref{eq:R3}.}\label{fig:R}
\end{figure}

\textit{Eighth step.} Having proved that the expression \eqref{eq:fr} gives us the desired family of functions, we now simplify this expression 
and transform it into a proper integral. To do it, we need the following standard theorem that can be found, for example, in \cite{budak}.
\begin{thrm}
Let $F(t,\theta)$ be a continuous function of two arguments on the set $[t_0, +\infty) \times [\theta_0, \theta_1]$. If the integral
\begin{equation}\label{eq:Fu}
    \tilde{F}(\theta) = \int^{+\infty}_{t_0} F(t, \theta) \, dt
\end{equation}
converges uniformly on the interval $[\theta_0, \theta_1]$, then the function $\tilde{F}(\theta)$ is integrable
on this interval and the following equality is valid:
\begin{equation}
\begin{split}
    \int^{\theta_1}_{\theta_0} \tilde{F}(\theta) \, d\theta &\equiv \int^{\theta_1}_{\theta_0} d\theta 
    \int^{+\infty}_{t_0} F(t, \theta) \, dt\\
    &= \int^{+\infty}_{t_0}dt\int^{\theta_1}_{\theta_0} F(t, \theta) \, d\theta.
\end{split}
\end{equation}
In other words, under these conditions one can change the order of integration.
\end{thrm}
Note that the Bessel function $I_0(z)$ can be represented as follows \cite{gr-I}:
\begin{equation}
    I_0(z) = \frac{1}{\pi} \int^{\pi}_0 e^{-z \cos\theta} \, d\theta.
\end{equation}
The Theorem 2 can be used to substitute this expression into the equation \eqref{eq:fr} and exchange the
order of integration. In fact, we have 
\begin{equation}
    f(r) = \frac{1}{\sqrt{2\pi K(\xi)}} e^{r/2} \int^{+\infty}_0 dt \int^{\pi}_0 F(t, \theta) \, d\theta,
\end{equation}
where
\begin{equation}
    F(t, \theta) = e^{-(1+\sqrt{\xi}\cos\theta)t}J_0(2\sqrt{rt}).
\end{equation}
To show that the integral \eqref{eq:Fu} converges uniformly on the interval $[0, \pi]$ for any
fixed $0 < \xi < 1$ note that
\begin{equation}
    F(t, \theta) = e^{-(1+\cos\theta)\sqrt{\xi}t} e^{-(1-\sqrt{\xi})t} J_0(2\sqrt{rt}).
\end{equation}
The function $e^{-(1-\sqrt{\xi})t} J_0(2\sqrt{rt})$ is integrable on $[0, +\infty)$ 
due to the equality \cite{gr-J0}
\begin{equation}
    \int^{+\infty}_0 e^{-at} J_0(b\sqrt{t}) \, dt = \frac{e^{-b^2/(4a)}}{a},
\end{equation}
valid for $a>0$. For $a=0$ this function is not integrable. We have that the integral
\begin{equation}
    \int^{+\infty}_0 e^{-(1-\sqrt{\xi})t} J_0(2\sqrt{rt}) \, dt
\end{equation}
converges uniformly with respect to $\theta$, since it does not depend on $\theta$, and 
the function $e^{-(1+\cos\theta)\sqrt{\xi}t}$ is uniformly bounded by 1. We conclude 
that the integral \eqref{eq:Fu} converges uniformly and exchanging the order of integration
is legal. We then obtain the following expression for the function $f(r) \equiv f_\xi(r)$:
\begin{equation}\label{eq:fe}
\begin{split}
    f_\xi(r) &= \frac{1}{\sqrt{2\pi K(\xi)}} e^{r/2} \int^{\pi}_0 d\theta \int^{+\infty}_0 F(t, \theta) \, dt \\
    &=\frac{1}{\sqrt{2\pi K(\xi)}} \int^{\pi}_0 
    \frac{\exp\left(-\frac{1-\sqrt{\xi}\cos\theta}{1+\sqrt{\xi}\cos\theta}\frac{r}{2}\right)}{1+\sqrt{\xi}\cos\theta} \, d\theta.
\end{split}
\end{equation}
This integral (in fact, even more general one) can be found in tables \cite{gr-exp}, and we have
\begin{equation}\label{eq:fxi}
\begin{split}
    f_\xi(r) = \sqrt{\frac{\pi}{2K(\xi)(1-\xi)}} I_0\left(\frac{\sqrt{\xi}}{1-\xi}r\right) e^{-\frac{1+\xi}{1-\xi}\frac{r}{2}}.
\end{split}
\end{equation}
Let us show that the this function has finite right derivative at the origin. We differentiate the expression \eqref{eq:fxi} with respect to $r$
and take the limit of this derivative for $r \to 0$. This limit is then the right derivative $f'_+(0)$. We have
\begin{equation}
    f'_+(0) = -\sqrt{\frac{\pi}{8K(\xi)}} \frac{1+\xi}{\sqrt{(1-\xi)^3}},
\end{equation}
and this expression is finite for all $0 < \xi < 1$.

In this step we finally obtain the expression for the family of wave functions:
\begin{equation}\label{eq:psixi}
\begin{split}
    &\psi(x, y) = \Psi(\sqrt{x^2+y^2}) = \frac{1}{\sqrt{\pi}} f(x^2+y^2) \\
               &= \frac{1}{\pi \sqrt{2K(\xi)}}
               \int^{\pi}_0 \frac{\exp\left(-\frac{1-\sqrt{\xi}\cos\theta}{1+\sqrt{\xi}\cos\theta}\frac{x^2+y^2}{2}\right)}{1+\sqrt{\xi}\cos\theta}\, d\theta \\
               &= \frac{\exp\left(-\frac{1+\xi}{1-\xi}\frac{x^2+y^2}{2}\right)}{\sqrt{2K(\xi)(1-\xi)}} I_0\left(\frac{\sqrt{\xi}}{1-\xi}(x^2+y^2)\right).
\end{split}
\end{equation}
From the finiteness of the right derivative $f'_+(0)$ we can conclude that this wave function is smooth at the origin, since the derivative $df(r^2)/dr = 2r f'(r^2)$ 
is equal to zero for $r = 0$. The function $\Psi(r)$, defined by Eq.~\eqref{eq:psip2}, is shown in Fig.~\ref{fig:psi}. In fact, it has meaning only for 
$0 \leqslant r < +\infty$, but here it is illustrated for the whole real line. Note that the shape of these functions is as we have expected.

\begin{figure}[h]
\includegraphics[scale=0.8]{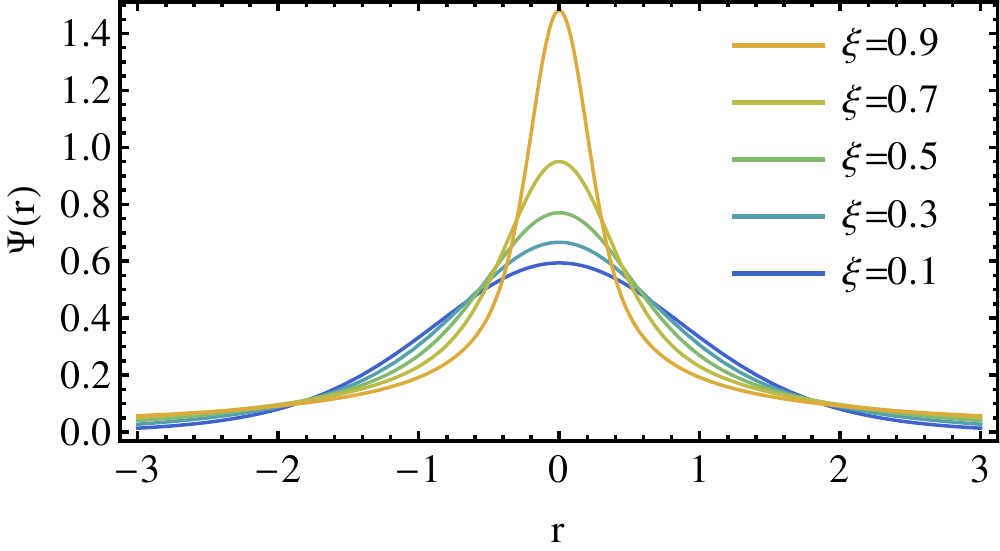}
\caption{The wave function $\Psi(r)$, defined by Eqs.~\eqref{eq:psip2} and \eqref{eq:fxi}, for a few different values of $\xi$.}\label{fig:psi}
\end{figure}

Using the integral representation of the wave function \eqref{eq:psixi}, we can compute the product $\sigma_{xp}\sigma_{px}$ according to the 
first equality of Eq.~\eqref{eq:s2}. Exchanging the orders of integration as we did before, we can arrive to the following expression for this product:
\begin{equation}
    \frac{1}{8\pi K(\xi)} \int^\pi_0 \int^\pi_0 \frac{(1-\xi \cos^2\theta)(1-\xi \cos^2\theta')}{(1-\xi \cos\theta \cos\theta')^3}
    \, d\theta \, d\theta'.
\end{equation}
This expression can be shown to agree with Eqs.~\eqref{eq:sigmaR} and \eqref{eq:R3}.
In this way we get directly from the definition the desired property of this wave function.
The bipartite state with the wave function \eqref{eq:psixi} we denote as $|\Psi_\xi\rangle$.
For all $0 < \xi < 1$ these states violate the inequality \eqref{eq:xppx} and when $\xi \to 1$, we have $\sigma_{xp} \sigma_{px} \to 1/8$, so the violation tends to 2.
In next sections, we study some simple properties of the states constructed and then extend this approach to a more general multipartite case.

\section{Properties} \label{sec:p}

In this section we compute the scalar product of the states $|\Psi_\xi\rangle$ with different values of the parameters $\xi$ and
find the decomposition of these states in the Fock basis. 

\subsection{Scalar product of states with different parameters} 

On can easily verify that the scalar product $\langle\Psi_\xi|\Psi_{\xi'}\rangle$ of the states $|\Psi_{\xi}\rangle$ and $|\Psi_{\xi'}\rangle$ 
can be computed as the scalar product of their corresponding functions $f_\xi(r)$ and $f_{\xi'}(r)$
\begin{equation}
    \langle\Psi_\xi|\Psi_{\xi'}\rangle = \int^{+\infty}_0 f_\xi(r) f_{\xi'}(r) \, dr.
\end{equation}

For the latter we have
\begin{equation}\label{eq:norm2}
\begin{split}
    &\int^{+\infty}_0 f_\xi(r) f_{\xi'}(r) \, dr = \frac{1}{2\pi \sqrt{K(\xi)K(\xi')}} 
    \int^{+\infty}_0 dr  \\
    &\int^{\pi}_0 d\theta \int^{\pi}_0 \frac{\exp\left(-\frac{1-\sqrt{\xi}\cos\theta}{1+\sqrt{\xi}\cos\theta}\frac{r}{2}
    -\frac{1-\sqrt{\xi'}\cos\theta'}{1+\sqrt{\xi'}\cos\theta'}\frac{r}{2}\right)}
    {(1+\sqrt{\xi}\cos\theta)(1+\sqrt{\xi'}\cos\theta')}
    \, d\theta'.
\end{split}
\end{equation}
In this expression the integrations over the angles $\theta$ and $\theta'$ are performed first and then the integration over $r$. Note that
\begin{equation}
\begin{split}
    &\frac{1}{2}\left(\frac{1-\sqrt{\xi}\cos\theta}{1+\sqrt{\xi}\cos\theta} 
    + \frac{1-\sqrt{\xi'}\cos\theta'}{1+\sqrt{\xi'}\cos\theta'}\right) \\
    &= \frac{1-\sqrt{\xi\xi'}\cos\theta\cos\theta'}{(1+\sqrt{\xi}\cos\theta)(1+\sqrt{\xi'}\cos\theta')} \\
    &\geqslant \frac{1-\sqrt{\xi\xi'}}{(1+\sqrt{\xi})(1+\sqrt{\xi'})} > 0.
\end{split}
\end{equation}
From this we can we can conclude that the order of the integrations can be changed and the integration over $r$ can be performed first and we get the equality
\begin{equation}
\begin{split}
    \int^{+\infty}_0 &f_\xi(r)f_{\xi'}(r) \, dr = \frac{1}{2\pi \sqrt{K(\xi)K(\xi')}} \\
    &\int^{\pi}_0 \int^{\pi}_0 \frac{d\theta \, d\theta'}{1 - \sqrt{\xi\xi'} \cos\theta \cos\theta'}.
\end{split}
\end{equation}
After the integration over $\theta$ we get
\begin{equation}
\begin{split}
    \int^{+\infty}_0 &f_\xi(r)f_{\xi'}(r) \, dr = \frac{1}{2\sqrt{K(\xi)K(\xi')}} \\
    &\int^{\pi}_0 \frac{d\theta'}{\sqrt{1-\xi\xi' \cos^2\theta'}}.
\end{split}
\end{equation}
Making substitution $z = \cos\theta'$, the last integral is transformed to the following one:
\begin{equation}\label{eq:norm3}
    \int^1_{-1} \frac{dz}{\sqrt{(1-\xi\xi' z^2)(1-z^2)}} = 2\int^1_0 \frac{dz}{\sqrt{(1-\xi\xi' z^2)(1-z^2)}},
\end{equation}
which, in turn, is transformed to $2K(\sqrt{\xi\xi'})$ by the substitution $z = \sin\theta$. We then obtain the desired scalar product
\begin{equation}
    \langle\Psi_\xi|\Psi_{\xi'}\rangle = \frac{K(\sqrt{\xi \xi'})}{\sqrt{K(\xi)K(\xi')}}.
\end{equation}
Note that for $\xi=\xi'$ we get the normalization condition $\langle\Psi_\xi|\Psi_\xi\rangle = 1$, as it must be.

\subsection{Fock states representation}

Now we derive the coefficients $c_{nm}$, $n,m \geqslant 0$, of the state $|\Psi_\xi\rangle$ in the Fock basis.
To do this, we need to find the scalar product of the state $|\Psi_\xi\rangle$ with the Fock state $|nm\rangle$, $c_{nm} = 
\langle nm|\Psi_\xi\rangle$ or, in other words, we need to compute the integral
\begin{equation}\label{eq:cnm}
    c_{nm} = \iint_{\mathbf{R}^2} \psi_n(x) \psi_m(y) \psi(x,y) \, dx \, dy, 
\end{equation}
where 
\begin{equation}
    \psi_n(x) = \frac{1}{\sqrt{\sqrt{\pi}2^nn!}} H_n(x) e^{-x^2/2}
\end{equation}
is the wave function of the Fock state $|n\rangle$. To compute this integral, we introduce the generating function
\begin{equation}\label{eq:F}
\begin{split}
    &F(x, y, u, v) = \psi_n(x) \psi_m(y) \frac{u^n v^m}{\sqrt{n! m!}} \\
            &= \frac{1}{\sqrt{\pi}} \exp\left(-\frac{x^2+y^2 + u^2 + v^2}{2} + \sqrt{2}(xu+yv)\right).
\end{split}
\end{equation}
This generating function is useful because it is much easier to compute the integral
\begin{equation}\label{eq:I2}
    I(u, v) = \iint_{\mathbf{R}^2} F(x, y, u, v) \psi(x, y) \, dx \, dy,
\end{equation}
and then expand it in $u$ and $v$ to find the coefficients $c_{nm}$. We have
\begin{equation}
\begin{split}
    I(u,v) &= \frac{e^{-\frac{u^2+v^2}{2}}}{\pi} \int^{+\infty}_0 r dr \\
           &\int^{2\pi}_0 f(r^2) e^{-r^2/2} e^{\sqrt{2}r(u \cos\theta + v \sin\theta)} \, d\theta.
\end{split}
\end{equation}
The inner integral (over $\theta$) can be easily taken by noting that $u \cos\theta + v \sin\theta = \sqrt{u^2+v^2}\cos(\theta+\theta_{u,v})$
and that due to periodicity of cosines function the shift $\theta_{u,v}$ plays no role in the integration over the period:
\begin{equation}
\begin{split}
    I(u,&v) = 2 e^{-\frac{u^2+v^2}{2}} \\
           &\int^{+\infty}_0 f(r^2) e^{-r^2/2} I_0(r\sqrt{2(u^2+v^2)})r \, dr = \\
           &e^{-\frac{u^2+v^2}{2}} \int^{+\infty}_0 f(r) e^{-r/2} I_0(\sqrt{2r(u^2+v^2)}) \, dr.
\end{split}
\end{equation}
We can substitute the expression \eqref{eq:fe} for the function $f(r)$ and get a repeated integral for $I(u, v)$. 
One can easily check, as it has been done before, that the order of integration can be changed and at the end we get
the following result:
\begin{equation}\label{eq:I}
    I(u,v) = \frac{1}{\sqrt{2\pi K(\xi)}} \int^{\pi}_0 e^{\frac{u^2+v^2}{2}\sqrt{\xi}\cos\theta} \, d\theta.
\end{equation}
The integral can be taken explicitly and expressed in terms of the Bessel function $I_0$, but the expression \eqref{eq:I}
is more convenient for our purpose. The exponent under the integral is the product of two exponents, one containing $u$ and the other
containing $v$. We can expand them in $u$ and $v$ respectively, multiply and integrate over $\theta$.
As one can see, we have to integrate powers of the cosines function. For these powers we have (see Eq.~\eqref{eq:sincos} below)
\begin{equation}
    \int^{\pi}_0 \cos^n \theta \, d\theta = \frac{\pi}{2^n} \binom{n}{n/2}
\end{equation}
if $n$ is even and the integral of odd powers are zero. From this, one can obtain
\begin{equation}\label{eq:I3}
    I(u,v) = \sqrt{\frac{\pi}{2 K(\xi)}} \sideset{}{'}\sum^{+\infty}_{n,m=0} \left(\frac{\sqrt{\xi}}{4}\right)^{n+m} 
    \binom{n+m}{\frac{n+m}{2}} \frac{u^{2n} v^{2m}}{n!m!},
\end{equation}
where the sum runs over all $n$ and $m$ with $n+m$ even. On the other hand, from Eqs.~\eqref{eq:cnm}, \eqref{eq:F} and \eqref{eq:I2} we have
\begin{equation}
    I(u, v) = \sum^{+\infty}_{n,m=0} c_{nm} \frac{u^n v^m}{\sqrt{n!m!}}.
\end{equation}
Comparing this with the expansion \eqref{eq:I3}, we obtain the coefficients $c_{nm}$
\begin{equation}\label{eq:cnm2}
    c_{nm} = \sqrt{\frac{\pi}{2 K(\xi)}} \sqrt{\binom{n}{\frac{n}{2}}\binom{m}{\frac{m}{2}}} \binom{\frac{n+m}{2}}{\frac{n+m}{4}} \left(\frac{\xi}{16}\right)^{\frac{n+m}{4}},
\end{equation}
where either $n = 4n'$, $m = 4m'$ or $n = 4n'+2$, $m = 4m'+2$, all other coefficients are zero. As an additional check, in Appendix \ref{app:cnm} we derive
directly from the expression \eqref{eq:cnm2} that the coefficients $c_{nm}$ satisfy the normalization condition.

\section{Generalization}\label{sec:g}

In this section we generalize our construction to an arbitrary even number $N = 2n$ of subsystems. We use the notation $\sigma^{(N)}_{xp} = \sigma_{x \ldots xp \ldots p}$
for the standard deviation of the operator $\hat{x}_1 \ldots \hat{x}_n\hat{p}_{n+1} \ldots \hat{p}_{2n}$ and, similarly, $\sigma^{(N)}_{px} = \sigma_{p \ldots px \ldots x}$
stands for the standard deviation of the operator $\hat{p}_1 \ldots \hat{p}_n\hat{x}_{n+1} \ldots \hat{x}_{2n}$. The following equalities take place:

\begin{equation}
\begin{split}
    \sigma^{(N)2}_{xp} &= \int_{\mathbf{R}^N} x^2_1 \ldots x^2_n \left(\frac{\partial^n \psi}{\partial x_{n+1} \ldots \partial x_N}\right)^2 \, d\vec{x}, \\
    \sigma^{(N)2}_{px} &= \int_{\mathbf{R}^N} x^2_{n+1} \ldots x^2_N \left(\frac{\partial^n \psi}{\partial x_1 \ldots \partial x_n}\right)^2 \, d\vec{x}.
\end{split}
\end{equation}
Here we again consider only spherically symmetric wave functions $\psi(x_1, \ldots, x_N) = \Psi(r)$, where, as before, we define
$r = \|\vec{x}\| = \sqrt{x^2_1 + \ldots + x^2_N}$. 
For the partial derivatives we have
\begin{equation}
    \frac{\partial \psi}{\partial x_i} = x_i \left(\frac{1}{r}\frac{d}{dr}\right) \Psi,
\end{equation}
and, in general, for distinct indices $i_1$, \ldots, $i_k$ we have
\begin{equation}
    \frac{\partial^k \psi}{\partial x_{i_1} \ldots \partial x_{i_k}} = x_{i_1} \ldots x_{i_k} \left(\frac{1}{r}\frac{d}{dr}\right)^k \Psi.
\end{equation}
For a spherically symmetric wave function Eq.~\eqref{eq:sZ} can be generalized as follows:
\begin{equation}
    \sigma^{(N)}_{xp}\sigma^{(N)}_{px} = \langle\hat{Z}^{(N)}\rangle,
\end{equation}
where the right-hand side is given by the integral
\begin{equation}
    \langle\hat{Z}^{(N)}\rangle = \int_{\mathbf{R}^N} x^2_1 \ldots x^2_N \left(\left(\frac{1}{r}\frac{d}{dr}\right)^n \Psi(r)\right)^2 \, d\vec{x}.
\end{equation}
To evaluate this integral we the spherical coordinates given by
\begin{equation}
\begin{split}
    x_1 &= r \cos\varphi_1, \\
    x_2 &= r \sin\varphi_1 \cos\varphi_2, \\
    x_3 &= r \sin\varphi_1 \cos\varphi_2 \cos\varphi_3, \\
     & \ldots \\
    x_{N-1} &= r \sin\varphi_1 \sin\varphi_2 \ldots \sin\varphi_{N-2} \cos\varphi_{N-1}, \\
    x_N &= r \sin\varphi_1 \sin\varphi_2 \ldots \sin\varphi_{N-2} \sin\varphi_{N-1},
\end{split}
\end{equation}
where $\varphi_1, \ldots, \varphi_{N-2} \in [0, \pi)$ and $\varphi_{N-1} \in [0, 2\pi)$. The Jacobian of this variable transformation is
\begin{equation}
    J = r^{N-1} \sin^{N-2}\varphi_1 \sin^{N-3}\varphi_2 \ldots \sin^2\varphi_{N-3} \sin\varphi_{N-2}.
\end{equation}
We then have the following equality:
\begin{equation}
\begin{split}
    &x^2_1 \ldots x^2_N J = r^{3N-1} \cos^2\varphi_1 \sin^{2(N-1)+N-2}\varphi_1 \\
    & \cos^2\varphi_2 \sin^{2(N-2)+N-3}\varphi_2 \ldots \cos^2\varphi_{N-1}\sin^2\varphi_{N-1}.
\end{split}
\end{equation}
The integrals over the angles can be easily taken explicitly \cite{gr-B} 
\begin{equation}\label{eq:sincos}
    \int^\pi_0 \sin^n\varphi \cos^m\varphi \, d\varphi = B\left(\frac{n+1}{2}, \frac{m+1}{2}\right),
\end{equation}
where $B(x, y)$ is the Euler beta-function. We can write
\begin{equation}\label{eq:b}
\begin{split}
    \int_\Omega x^2_1 &\ldots x^2_N J \, d\pmb{\varphi} = B\left(\frac{3}{2},\frac{3(N-1)}{2}\right) \\
    &B\left(\frac{3}{2},\frac{3(N-2)}{2}\right) \ldots B\left(\frac{3}{2},\frac{3 \cdot 2}{2}\right) \frac{\pi}{4},
\end{split}
\end{equation}
where $\Omega = [0, \pi)^{N-2} \times [0, 2\pi)$. Using the well-known relation for the beta-function
\begin{equation}
    B(x, y) = \frac{\Gamma(x)\Gamma(y)}{\Gamma(x+y)},
\end{equation}
the right-hand side of Eq.~\eqref{eq:b} can be simplified as follows:
\begin{equation}
    \int_\Omega x^2_1 \ldots x^2_N J \, d\pmb{\varphi} = \frac{\sqrt{\pi^N}}{2^{N-1}\Gamma\left(\frac{3N}{2}\right)}.
\end{equation}
This expression is valid for all $N \geqslant 1$, including odd values. In our case $N = 2n$ is even and this expression can be further simplified:
\begin{equation}
    \int_\Omega x^2_1 \ldots x^2_{2n} J \, d\pmb{\varphi} = \frac{\pi^n}{2^{2n-1}(3n-1)!}.
\end{equation}
Finally, for the averaged value $\langle\hat{Z}^{(N)}\rangle$ we have 
\begin{equation}\label{eq:Zn}
\begin{split}
    \langle\hat{Z}^{(N)}\rangle &= \frac{\pi^n}{2^{2n-1}(3n-1)!} \\
    &\int^{+\infty}_0 r^{6n-1} \left(\left(\frac{1}{r}\frac{d}{dr}\right)^n \Psi(r)\right)^2 \, dr.
\end{split}
\end{equation}
For $n = 1$ this expression coincides with Eq.~\eqref{eq:Z2}. From the normalization condition
\begin{equation}
    \int_{\mathbf{R}^N} \psi^2(x_1, \ldots, x_N) \, d\vec{x} = 1
\end{equation}
in the same way we can obtain the following equality:
\begin{equation}\label{eq:fn}
    \frac{2\pi^n}{(n-1)!} \int^{+\infty}_0 r^{2n-1} \Psi^2(r) \, dr = 1.
\end{equation}
For $n=1$ it coincides with Eq.~\eqref{eq:psip}. Let us introduce the function $f(r)$ via the relation
\begin{equation}\label{eq:fn2}
    \Psi(r) = \sqrt{\frac{n!}{\pi^n}} f(r^{2n}),
\end{equation}
which generalizes Eq.~\eqref{eq:psip2}. From Eq.~\eqref{eq:fn} we have
\begin{equation}
\begin{split}
    &\frac{2\pi^n}{(n-1)!}\frac{n!}{\pi^n} \int^{+\infty}_0 f^2(r^{2n}) r^{2n-1} \, dr \\
    &= \int^{+\infty}_0 f^2(r^{2n}) \, dr^{2n} = \int^{+\infty}_0 f^2(r) \, dr = 1.
\end{split}
\end{equation}
We see that the function $f(r)$ is normalized in the ordinary sense. Computation of the integral \eqref{eq:Zn}
requires more work. As one can easily check, we have
\begin{equation}\label{eq:a}
   \left(\frac{1}{r}\frac{d}{dr}\right)^n f(r^{2n}) = \sum^n_{k=1} a_k r^{2n(k-1)} f^{(k)}(r^{2n}),
\end{equation}
where the coefficients $a_k$ are to be determined. We can write the Taylor expansion of the function $f(r^{2n})$ as
\begin{equation}
    f(r^{2n}) = \sum^{+\infty}_{n=0} \frac{f^{(k)}(0)}{k!} r^{2nk},
\end{equation}
and apply the differential operator to this expansion
\begin{equation}\label{eq:a1}
\begin{split}
    \left(\frac{1}{r}\frac{d}{dr}\right)^n f(r^{2n}) &= \sum^{+\infty}_{k=1} \frac{f^{(k)}(0)}{k!} \frac{2^n (nk)!}{(n(k-1))!} r^{2n(k-1)} \\
    &= 2^n \sum^{+\infty}_{k=0} \frac{f^{(k+1)}(0)}{(k+1)!} \frac{n(k+1)!}{(nk)!}r^{2nk}.
\end{split}
\end{equation}
On the other hand, each of the $n$ derivatives on the right-hand side of Eq.~\eqref{eq:a} can also be expanded into its Taylor series 
and when these expansions are substituted into Eq.~\eqref{eq:a} we get
\begin{equation}\label{eq:a2}
    \sum^n_{k=1} \sum^{+\infty}_{j=0} a_k \frac{f^{k+j}(0)}{j!} r^{2n(k+j-1)}.
\end{equation}
The two expressions \eqref{eq:a1} and \eqref{eq:a2} must agree, so the coefficients of the powers $r^{2nk}$ must be equal 
for all $k = 0, 1, \ldots$. When we equate these coefficients for $k = 0, 1, \ldots, n-1$ we get the following system of equations:
\begin{equation}
    \sum^m_{k=1} \frac{a_k}{(m-k)!} = 2^n \frac{(nm)!}{m!(n(m-1))!} = \frac{2^n n!}{m!}\binom{nm}{n},
\end{equation}
for $m = 1, \ldots, n$. This system has the following matrix:
\begin{equation}\label{eq:A}
    A = 
    \begin{pmatrix}
        \frac{1}{0!} & 0 & 0 & \ldots & 0 \\
        \frac{1}{1!} & \frac{1}{0!} & 0 & \ldots & 0 \\
        \frac{1}{2!} & \frac{1}{1!} & \frac{1}{0!} & \ldots & 0 \\
        \hdotsfor{5} \\
        \frac{1}{(n-1)!} & \frac{1}{(n-2)!} & \frac{1}{(n-3)!} & \ldots & \frac{1}{0!}
    \end{pmatrix}.
\end{equation}
One can easily verify that the inverse matrix $A^{-1}$ is given by a similar expression
\begin{equation}\label{eq:A-1}
    A^{-1} = 
    \begin{pmatrix}
        \frac{1}{0!} & 0 & 0 & \ldots & 0 \\
        -\frac{1}{1!} & \frac{1}{0!} & 0 & \ldots & 0 \\
        \frac{1}{2!} & -\frac{1}{1!} & \frac{1}{0!} & \ldots & 0 \\
        \hdotsfor{5} \\
        \frac{(-1)^{n-1}}{(n-1)!} & \frac{(-1)^{n-2}}{(n-2)!} & \frac{(-1)^{n-3}}{(n-3)!} & \ldots & \frac{1}{0!}
    \end{pmatrix}.
\end{equation}
In fact, the elements of these matrices are
\begin{equation}
    A_{ij} = 
    \begin{cases}
        0 & i < j \\
        \frac{1}{(i-j)!} & i \geqslant j
    \end{cases}, \quad
    (A^{-1})_{ij} = 
    \begin{cases}
        0 & i < j \\
        \frac{(-1)^{i+j}}{(i-j)!} & i \geqslant j
    \end{cases}.
\end{equation}
The matrix element of the product is
\begin{equation}
    (A \cdot A^{-1})_{ij} = \sum^{n-1}_{k=0} A_{ik} (A^{-1})_{kj}.
\end{equation}
Only the terms with $i \geqslant k$ and $k \geqslant j$ are nonzero in this sum, so all terms above the main diagonal, i.e. the terms with $i < j$, are zero 
(the product of two lower triangular matrices is also lower triangular).
For $i \geqslant j$ we have
\begin{equation}
\begin{split}
    (A \cdot A^{-1})_{ij} &= \sum^i_{k=j} \frac{1}{(i-k)!}\frac{(-1)^{k+j}}{(k-j)!} \\
    &= \frac{1}{(i-j)!} \sum^{i-j}_{k=0} (-1)^k \binom{i-j}{k} = \delta_{ij},
\end{split}
\end{equation}
where the last equality is due to the well-known identity for the binomial coefficients \cite{gr-bin1}. We see that the product
of these two matrices is the identity matrix,
$A \cdot A^{-1} = E_n$, so the matrix $A^{-1}$ defined by Eq.~\eqref{eq:A-1} is indeed the inverse of the matrix \eqref{eq:A}.

We now can express the coefficients $a_k$ explicitly
\begin{equation}
\begin{split}
    a_k &= \sum^k_{j=1} \frac{(-1)^{j+k}}{(k-j)!} 2^n n! \frac{1}{j!} \binom{jn}{n} \\
        &= 2^n n! \frac{(-1)^k}{k!} \sum^k_{j=1} (-1)^j \binom{k}{j} \binom{j n}{n}.
\end{split}
\end{equation}
Now one can verify that coefficients of the powers $r^{2nk}$ in the expression \eqref{eq:a1} and \eqref{eq:a2} agree not only for $k = 0, 1, \ldots, n-1$,
but for all $k$.
When we substitute Eqs.~\eqref{eq:fn2}, \eqref{eq:a} into Eq.~\eqref{eq:Zn} we get an explicit expression for $\langle\hat{Z}^{(n)}\rangle$
in the form of the following functional:
\begin{equation}\label{eq:Zn2}
    \langle\hat{Z}^{(n)}\rangle = 3\frac{(n!)^3}{(3n)!} \int^{+\infty}_0 \left(\sum^n_{k=1}b_k r^k f^{(k)}(r)\right)^2 \, dr,
\end{equation}
where the coefficients $b_k$ read as
\begin{equation}\label{eq:b2}
    b_k = \frac{(-1)^k}{k!} \sum^k_{j=1} (-1)^j \binom{k}{j} \binom{j n}{n}.
\end{equation}
For $n=1$ Eq.~\eqref{eq:Zn2} coincides with Eq.~\eqref{eq:rf}.

First, we show that the integral on the right-hand side of Eq.~\eqref{eq:Zn2} is never zero for a square-integrable function $f(r)$.
In fact, if this integral is zero then $f(r)$ satisfies the following differential equation of the $n$th order on the interval $(0, +\infty)$:
\begin{equation}\label{eq:df}
    \sum^n_{k=1}b_k r^k f^{(k)}(r) = 0.
\end{equation}
This equation has $n$ linearly independent solutions. We can try to find them in the form $f(r) = r^\alpha$, where $\alpha$ must satisfy the equation
\begin{equation}\label{eq:alpha-eq}
    \sum^n_{k=1}b_k (\alpha)_k = 0,
\end{equation}
where $(\alpha)_k = \alpha(\alpha-1)\ldots(\alpha-k+1)$ is the Pochhammer symbol. One can easily verify that this equation has $n$ different roots
\begin{equation}\label{eq:alpha}
    \alpha = \frac{0}{n}, \frac{1}{n}, \ldots, \frac{n-1}{n}.
\end{equation}
This follows from the fact that if we take $f(r) = r^{j/n}$, then $f(r^{2n}) = r^{2k}$ and 
\begin{equation}
    \left(\frac{1}{r}\frac{d}{dr}\right)^n f(r^{2n}) = \left(\frac{1}{r}\frac{d}{dr}\right)^n r^{2k} = 0
\end{equation}
for $k = 0, 1, \ldots, n-1$. In Appendix \ref{app:roots} we prove directly from the expression
for the coefficients $b_k$, Eq.~\eqref{eq:b2}, that the numbers \eqref{eq:alpha} satisfy the equation \eqref{eq:alpha-eq}. 
We see that any solution of the equation \eqref{eq:df} has the form
\begin{equation}
    f(r) = C_{n-1} \sqrt[n]{r^{n-1}} + \ldots + C_1 \sqrt[n]{r} + C_0.
\end{equation}
None of these functions is normalizable on the interval $(0, +\infty)$.

Now we find the infimum of the operators $\langle\hat{Z}^{(n)}\rangle$ over the set of all normalized functions $f(r)$. It is rather a challenging task to
do it in general so we find the minimum for $n = 2, 3$ and develop a method which allows to do it in general. In both cases the functions obtained in the previous step
are used. In contrast to the bipartite case, here we could not find explicit analytical expressions for the solutions we obtain, but we developed 
a technique to prove that these solution are indeed optimal without having explicit expressions for them.

\subsection{Four-partite case}

For $N = 4$ we need to minimize the integral
\begin{equation}
    \langle\hat{Z}^{(4)}\rangle = \frac{1}{30} \int^{+\infty}_0 (r f'(r) + 2r^2 f^{\prime\prime}(r))^2 \, dr.
\end{equation}
We prove a more general
\begin{thrm}\label{thm:n2}
The equality 
\begin{equation}\label{eq:fa2}
    \inf_{\|f\| = 1} \|r f'+a r^2 f^{\prime\prime}\|^2 = \frac{1}{4}\left(\frac{3}{2}a-1\right)^2,
\end{equation}
is valid provided that $a \geqslant 1$. 
\end{thrm}
Another formulation of this theorem is the statement that for an arbitrary non-normalized (but normalizable) function $f$ we have the inequality
\begin{equation}\label{eq:fa2a}
    \|r f'+a r^2 f^{\prime\prime}\| > \frac{1}{2}\left(\frac{3}{2}a-1\right) \|f\|,
\end{equation}
and the coefficient on the right-hand side is the best possible (for $a \geqslant 1$).
\begin{proof}
Computing the scalar product
\begin{equation}
    (rf' + a r^2 f^{\prime\prime}, rf') = \left(1-\frac{3}{2}a\right) \|r f'\|^2,
\end{equation}
we can apply Cauchy-Schwarz inequality to get
\begin{equation}
    \|rf' + a r^2 f^{\prime\prime}\| \geqslant \left(\frac{3}{2}a-1\right) \|r f'\| > \frac{1}{2} \left(\frac{3}{2}a-1\right),
\end{equation}
where we used the inequality \eqref{eq:rff} and the normalization $\|f\| = 1$. We need to show that one can find such functions $f$ with $\|f\|=1$ that 
the difference between the left-hand and the right-hand side of the inequality \eqref{eq:fa2} becomes 
arbitrary small. We construct such functions with the help of the functions we obtained for $N = 2$. 

Let us take the function $f_\xi(r) \equiv f^{(2)}_\xi(r)$ 
defined by Eq.~\eqref{eq:fxi} and find a normalizable solution $g_\xi(r)$ of the equation
\begin{equation}\label{eq:g}
    (1-a)g_\xi + a r g'_\xi = f_\xi.
\end{equation}
To solve it, we first solve the following auxiliary equation, motivated by the integral representation \eqref{eq:fe}:
\begin{equation}\label{eq:G}
    (1-a)G + a r G' = e^{-\gamma r}, \quad \gamma = \frac{1}{2}\frac{1-\sqrt{\xi}\cos\theta}{1+\sqrt{\xi}\cos\theta}.
\end{equation}
Then, the corresponding solution of Eq.~\eqref{eq:g} is
\begin{equation}
    g_\xi(r) = \frac{1}{\sqrt{2\pi K(\xi)}} \int^\pi_0 \frac{G(r)}{1+\sqrt{\xi}\cos\theta} \, d\theta,
\end{equation}
where, in fact, $G(r)$ depends on both $\xi$ and $\theta$, since $\gamma$ in Eq.~\eqref{eq:G} depends on them.
The general solution of Eq.~\eqref{eq:G} is
\begin{equation}\label{eq:G2}
    G(r) = cr^{\frac{a-1}{a}} - \frac{1}{a} (\gamma r)^{\frac{a-1}{a}}\Gamma\left(-\frac{a-1}{a}, \gamma r\right),
\end{equation}
so a normalizable solution of Eq.~\eqref{eq:g} is given by
\begin{equation}\label{eq:g2}
    g_\xi(r) = -\frac{1}{a\sqrt{2\pi K(\xi)}} \int^\pi_0 \frac{(\gamma r)^{\frac{a-1}{a}}\Gamma\left(-\frac{a-1}{a}, \gamma r\right)}{1+\sqrt{\xi}\cos\theta} \, d\theta,
\end{equation}
where $\Gamma(\zeta, x)$ is the incomplete Gamma-function \cite{gr-gamma1}
\begin{equation}
    \Gamma(\zeta, x) = \int^{+\infty}_x e^{-t} t^{\zeta-1} \, dt.
\end{equation}
The solution \eqref{eq:G2} for $c=0$ can be also written as \cite{gr-gamma2}
\begin{equation}\label{eq:G3}
    G(r) = -\frac{e^{-\gamma r}}{a \Gamma\left(2-\frac{1}{a}\right)} \int^{+\infty}_0 \frac{e^{-t} t^{\frac{a-1}{a}}}{\gamma r + t} \, dt,
\end{equation}
so it is clear that the corresponding solution \eqref{eq:g2} is normalizable. We could not find an explicit expression for the 
function $g_\xi(r)$ itself or for its norm in terms of known 
special functions, but it is possible to find the limit of $\|g_\xi\|$ when $\xi \to 1$.

From Eq.~\eqref{eq:g} we obtain
\begin{equation}
    \quad r g'_\xi + a r^2 g^{\prime\prime}_\xi = r f'_\xi,
\end{equation}
and, since the solution $g_\xi(r)$ is normalizable, we have 
\begin{equation}\label{eq:h2}
\begin{split}
    (1-a)(1-2a)\|g_\xi\|^2 + a^2 \|r g'_\xi\|^2 &= \|f_\xi\|^2 = 1, \\
    (1-3a) \|r g'_\xi\|^2 + a^2 \|r^2 g^{\prime\prime}_\xi\|^2 &= \|r f'_\xi\|^2.
\end{split}
\end{equation}
To estimate $\|r^2 g^{\prime\prime}_\xi\|$, let us compute the scalar product
\begin{equation}
    (r^2 g^{\prime\prime}_\xi, r g'_\xi) = -\frac{3}{2} \|r g'_\xi\|^2,
\end{equation}
from which we get the inequality
\begin{equation}
    \|r^2 g^{\prime\prime}_\xi\| \geqslant \frac{3}{2} \|r g'_\xi\|,
\end{equation}
and from the second equality of Eq.~\eqref{eq:h2} we get
\begin{equation}
    \|r f'_\xi\|^2 \geqslant \left(1-3a+\frac{9}{4}a^2\right) \|r g'_\xi\|^2 = \left(\frac{3}{2}a-1\right)^2 \|r g'_\xi\|^2.
\end{equation}
Then, from this inequality, the first equality of Eq.~\eqref{eq:h2} and the inequality \eqref{eq:rff} we have
\begin{equation}
\begin{split}
    1 &\leqslant (4(1-a)(1-2a)+a^2) \|r g'_\xi\|^2 \\
    &= 4\left(\frac{3}{2}a-1\right)^2 \|r g'_\xi\|^2 \leqslant 4\|r f'_\xi\|^2.
\end{split}
\end{equation}
It is in this place that we use the condition $a \geqslant 1$ that
guarantees non-negativity of the coefficient $(1-a)(1-2a)$. 
By construction of the function $f_\xi$, we have $4\|r f'_\xi\|^2 \to 1$ when $\xi \to 1$, from which we immediately get 
\begin{equation}
    \|r g'_\xi\|^2 \to \frac{1}{4\left(\frac{3}{2}a-1\right)^2},
\end{equation}
and thus $\|g_\xi\|^2 \to \left(\frac{3}{2}a-1\right)^{-2}$ when $\xi \to 1$. Results of numerical integration of $g_\xi$ 
for several randomly chosen $a>1$ and for $\xi \approx 1$ agree with this conclusion (within the accuracy provided by the numerical integration routines). 
If we take the function
\begin{equation}\label{eq:gtilde}
    \tilde{g}_\xi(r) = \frac{g_\xi(r)}{\|g_\xi\|}
\end{equation}
then this function is normalized, $\|\tilde{g}_\xi\| = 1$ for all $\xi \in (0, 1)$, and we have
\begin{equation}
    \|r\tilde{g}'_\xi + a r^2 \tilde{g}^{\prime\prime}_\xi\|^2 = \frac{\|r f'_\xi\|^2}{\|g_\xi\|^2} \to \frac{1}{4}\left(\frac{3}{2}a-1\right)^2
\end{equation}
when $\xi \to 1$, which concludes the proof. 
\end{proof}

From this theorem, for $a = 2$, we get $\inf \langle\hat{Z}^{(4)}\rangle = 1/30$. As a minimizing family of functions we can take the functions 
$f^{(4)}_\xi(r) = \tilde{g}_\xi(r)$, where $\tilde{g}_\xi(r)$ is given by Eqs.~\eqref{eq:gtilde} and \eqref{eq:g2} for $a=2$. The corresponding
function $\Psi(r)$, Eq.~\eqref{eq:fn2}, is shown in Fig.~\ref{fig:psi2}.
The wave function $\psi(\vec{x})$ is given by
\begin{equation}
\begin{split}
    \psi(x_1, &x_2, x_3, x_4) = \Psi(r) = \frac{\sqrt{2}}{\pi} f^{(2)}_\xi(r^4) \\
    &= \frac{\sqrt{2}}{\pi} f^{(2)}_\xi((x^2_1+x^2_2+x^2_3+x^2_4)^2).
\end{split}
\end{equation}
The functions \eqref{eq:gtilde} for $a=3/2$, which we will denote by the same symbol $f^{(4)}_\xi(r)$, will be used in the next subsection
to construct a minimizing family of functions for the case $N=6$.

\begin{figure}[h]
\includegraphics[scale=0.8]{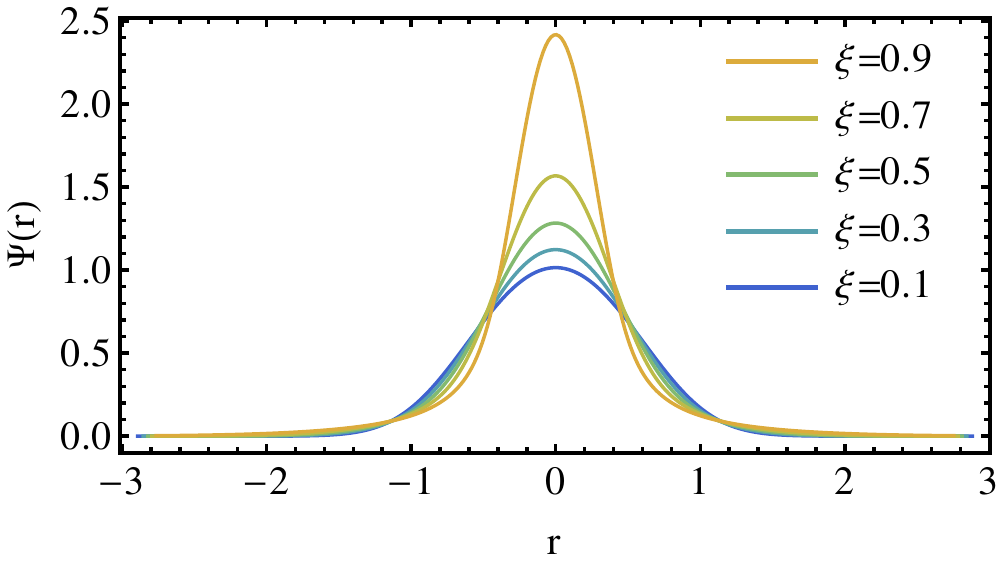}
\caption{The wave function $\Psi(r)$, defined by Eqs.~\eqref{eq:fn2} for $n=2$, \eqref{eq:gtilde} and \eqref{eq:g2} with $a=2$ for a few different values of $\xi$.}\label{fig:psi2}
\end{figure}

\subsection{Six-partite case}

For $N=6$ we have to minimize the integral
\begin{equation}\label{eq:ZZ3}
\begin{split}
    \langle\hat{Z}^{(6)}\rangle = \frac{1}{560} \int^{+\infty}_0 \Biggl(r f'(r) &+ 9r^2 f^{\prime\prime}(r) \Biggr. \\
    &+ \Biggl.\frac{9}{2}r^3 f^{\prime\prime\prime}(r)\Biggr)^2 \, dr.
\end{split}
\end{equation}
The solution is given by
\begin{thrm}
The following equality holds true:
\begin{equation}
    \inf \langle\hat{Z}^{(6)}\rangle = \frac{35}{4096}.
\end{equation}
A family of functions that approach this infimum can be constructed with the help of the functions obtained for the previous case $n=2$.
\end{thrm}
\begin{proof}
We have
\begin{equation}
\begin{split}
    r f' &+ 9r^2 f^{\prime\prime} + \frac{9}{2}r^3 f^{\prime\prime\prime} \\
    &= r(3rf'-2f)' + \frac{3}{2}r^2(3rf'-2f)^{\prime\prime},
\end{split}
\end{equation}
and from Eq.~\eqref{eq:fa2a} we get the inequality
\begin{equation}\label{eq:f3}
\begin{split}
    \|r f' + 9r^2 f^{\prime\prime} + \frac{9}{2}r^3 f^{\prime\prime\prime}\| &> \frac{1}{2}\left(\frac{3}{2}\cdot\frac{3}{2}-1\right)\|2f-3rf'\| \\
    &= \frac{5}{8}\|2f-3rf'\|.
\end{split}
\end{equation}
The value $\|2f-3rf'\|^2$ can be estimated as follows:
\begin{equation}
\begin{split}
    \|2f&-3rf'\|^2 = 10 \|f\|^2 + 9 \|rf'\|^2 \\
    &> \left(10+\frac{9}{4}\right)\|f\|^2 = \frac{49}{4}\|f\|^2 = \frac{49}{4},
\end{split}
\end{equation}
where we used the inequality \eqref{eq:rff}. Combining this result with the inequality \eqref{eq:f3} we get
\begin{equation}
    \|r f' + 9r^2 f^{\prime\prime} + \frac{9}{2}r^3 f^{\prime\prime\prime}\| > \frac{5}{8} \cdot \frac{7}{2} = \frac{35}{16}.
\end{equation}
We now prove that this estimation is the best possible.

Consider a normalizable solution of the equation
\begin{equation}\label{eq:h}
    -2h_\xi + 3r h'_\xi = f^{(4)}_\xi,
\end{equation}
where $f^{(4)}_\xi(r)$ is the function given by Eqs.~\eqref{eq:gtilde} and \eqref{eq:g2} for $a = 3/2$. This function $f^{(4)}_\xi(r)$ has the following integral representation:
\begin{equation}
    f^{(4)}_\xi(r) = N(\xi) \int^\pi_0 \frac{\sqrt[3]{\gamma r}\, \Gamma\left(-\frac{1}{3}, \gamma r\right)}{1+\sqrt{\xi}\cos\theta} \, d\theta,
\end{equation}
where $N(\xi)$ is the normalization such that $\|f^{(4)}_\xi\| = 1$. To solve Eq.~\eqref{eq:h} we first find a solution of the equation
\begin{equation}\label{eq:H2}
    -2H_\xi + 3r H'_\xi = \sqrt[3]{\gamma r}\, \Gamma\left(-\frac{1}{3}, \gamma r\right),
\end{equation}
and then construct from it the corresponding solution of Eq.~\eqref{eq:h} via
\begin{equation}\label{eq:hxi}
    h_\xi(r) = N(\xi) \int^\pi_0 \frac{H_\xi(r)}{1+\sqrt{\xi}\cos\theta} \, d\theta.
\end{equation}
A normalizable solution of Eq.~\eqref{eq:H2} reads as
\begin{equation}\label{eq:Hxi}
\begin{split}
    H_\xi(r) &= \frac{3}{2} e^{-\gamma r} - \frac{3}{2} \sqrt[3]{(\gamma r)^2} \, \Gamma\left(\frac{1}{3}, \gamma r \right) \\
    &- \sqrt[3]{\gamma r} \, \Gamma\left(-\frac{1}{3}, \gamma r \right),
\end{split}
\end{equation}
or, according to \cite{gr-gamma2}, this solution can also be written as
\begin{equation}
\begin{split}
    H_\xi(r) = \frac{3}{2} e^{-\gamma r} &- \frac{3}{2} \frac{\gamma r e^{-\gamma r}}{\Gamma\left(\frac{2}{3}\right)} 
    \int^{+\infty}_0 \frac{e^{-t}}{\sqrt[3]{t}(\gamma r + t)} \, dt \\
    &-\frac{e^{-\gamma r}}{\Gamma\left(\frac{4}{3}\right)} \int^{+\infty}_0 \frac{\sqrt[3]{t} e^{-t}}{\gamma r + t} \, dt,
\end{split}
\end{equation}
from which it is clear that this solution is normalizable.

We could not find an explicit expression for the function $h_\xi(r)$ itself or for its norm $\|h_\xi\|$ but we can prove
some statements about the behavior of the norm $\|h_\xi\|$ when $\xi \to 1$. From Eq.~\eqref{eq:h} we have
\begin{equation}\label{eq:fh4}
    rh'_\xi + 9r^2 h^{\prime\prime}_\xi + \frac{9}{2} r^3 h^{\prime\prime\prime}_\xi = r f^{(4)\prime}_\xi + \frac{3}{2} r^2 f^{(4)\prime\prime}_\xi.
\end{equation}
Taking the norm of the both sides of this equality, we get
\begin{equation}\label{eq:fh}
\begin{split}
    28 \|r h'_\xi\|^2 &- \frac{261}{2} \|r^2 h^{\prime\prime}_\xi\|^2 + \frac{81}{4} \|r^3 h^{\prime\prime\prime}_\xi\|^2 \\
    &= \left\|r f^{(4)\prime}_\xi + \frac{3}{2} r^2 f^{(4)\prime\prime}_\xi\right\|^2,
\end{split}
\end{equation}
where we used  the relations
\begin{equation}
\begin{split}
    (r h'_\xi, r^2 h^{\prime\prime}_\xi) &= -\frac{3}{2} \|r h'_\xi\|^2, \ (r^2 h^{\prime\prime}_\xi, r^3 h^{\prime\prime\prime}_\xi) = -\frac{5}{2} \|r^2 h^{\prime\prime}_\xi\|^2, \\
    (r h'_\xi, r^3 h^{\prime\prime\prime}_\xi) &= 6\|r h'_\xi\|^2 - \|r^2 h^{\prime\prime}_\xi\|^2.
\end{split}
\end{equation}
From the second of these relations we can get the following inequality:
\begin{equation}\label{eq:h23}
    \|r^3 h^{\prime\prime\prime}_\xi\| \geqslant \frac{5}{2} \|r^2 h^{\prime\prime}_\xi\|.
\end{equation}
If we substitute it into the left-hand side of Eq.~\eqref{eq:fh} we get
\begin{equation}
    28 \|r h'_\xi\|^2 - \frac{63}{16} \|r^2 h^{\prime\prime}_\xi\|^2,
\end{equation}
so the second coefficient is negative and we cannot follow the idea of the prove for $n=2$. Up to now we used relations like \eqref{eq:h23}
which are derived from Cauchy-Schwarz inequality and this inequality is based on the non-negativity of the expression $\|\alpha f + \beta g\|^2$
for all real numbers $\alpha$ and $\beta$. Here we need to go one step further and use the non-negativity of similar expression with three terms. Namely,
we can write
\begin{equation}
\begin{split}
    &\left\|\alpha r h'_\xi + \beta r^2 h^{\prime\prime}_\xi + \frac{9}{2} r^3 h^{\prime\prime\prime}_\xi\right\|^2 
    = (\alpha^2 - 3\alpha\beta + 54\alpha) \|r h'_\xi\|^2 \\ 
    &+ \left(\beta^2 - 9\alpha -\frac{45}{2}\beta\right) \|r^2 h^{\prime\prime}_\xi\|^2
    + \frac{81}{4} \|r^3 h^{\prime\prime\prime}_\xi\|^2 \geqslant 0.
\end{split}
\end{equation}
Then we find such $\alpha$ and $\beta$ that
\begin{equation}
\begin{split}
    \alpha^2 - 3\alpha\beta + 54\alpha &= 28-49\left(\frac{5}{8}\right)^2, \\
    \beta^2 - 9\alpha -\frac{45}{2}\beta &= -\frac{261}{2}.
\end{split}
\end{equation}
One can verify that this system has two solutions
\begin{equation}
    \alpha = \frac{9}{8}(9\pm\sqrt{74}), \quad \beta = \frac{3}{4}(24\pm\sqrt{74}).
\end{equation}
If we take one these two solutions, $(\alpha_0, \beta_0)$, the equality \eqref{eq:fh} can be written as
\begin{equation}
\begin{split}
    49\left(\frac{5}{8}\right)^2 \|r h'_\xi\|^2  &+ \left\|\alpha_0 r h'_\xi + \beta_0 r^2 h^{\prime\prime}_\xi + \frac{9}{2} r^2 h^{\prime\prime}_\xi\right\|^2 \\
    &= \left\|r f^{(4)\prime}_\xi + \frac{3}{2} r^2 f^{(4)\prime\prime}_\xi\right\|^2,
\end{split}
\end{equation}
so that we have the inequality
\begin{equation}\label{eq:fh2}
    49\left(\frac{5}{8}\right)^2 \|r h'_\xi\|^2 \leqslant \left\|r f^{(4)\prime}_\xi + \frac{3}{2} r^2 f^{(4)\prime\prime}_\xi\right\|^2.
\end{equation}
From the differential equation \eqref{eq:h} we have
\begin{equation}\label{eq:hhprime}
    10 \|h_\xi\|^2 + 9 \|r h'_\xi\|^2 = \|f^{(2)}_\xi\|^2 = 1
\end{equation}
for all $\xi$. Then, using Eq.~\eqref{eq:fh2}, we have
\begin{equation}\label{eq:fh3}
\begin{split}
    1 &= 10 \|h_\xi\|^2 + 9 \|r h'_\xi\|^2 \leqslant 49 \|r h'_\xi\|^2 \\
    &\leqslant \left(\frac{8}{5}\right)^2 \left\|r f^{(4)\prime}_\xi + \frac{3}{2} r^2 f^{(4)\prime\prime}_\xi\right\|^2.
\end{split}
\end{equation}
We have chosen the function $f^{(4)}_\xi(r)$ in such a way that
\begin{equation}
    \left\|r f^{(4)\prime}_\xi + \frac{3}{2} r^2 f^{(4)\prime\prime}_\xi\right\|^2 \to \frac{1}{4}\left(\frac{3}{2}\cdot\frac{3}{2}-1\right) = \left(\frac{5}{8}\right)^2
\end{equation}
when $\xi \to 1$. From Eq.~\eqref{eq:fh3} we get that $\|r h'_\xi\| \to 1/7$ and then from Eq.~\eqref{eq:hhprime} we obtain that $\|h_\xi\| \to 2/7$ when $\xi \to 1$.
Results of numerical integration for $\xi \approx 1$ agree with this conclusion.
Now, let us take the function
\begin{equation}\label{eq:htilde}
    \tilde{h}_\xi(r) = \frac{h_\xi(r)}{\|h_\xi\|}.
\end{equation}
This function is normalized, $\|\tilde{h}_\xi\| = 1$, and from Eq.~\eqref{eq:fh4} we get
\begin{equation}
\begin{split}
    \left\|r\tilde{h}'_\xi + 9r^2 \tilde{h}^{\prime\prime}_\xi + \frac{9}{2} r^3 \tilde{h}^{\prime\prime\prime}_\xi\right\| &=
    \frac{1}{\|h_\xi\|} \left\|r f^{(4)\prime}_\xi + \frac{3}{2} r^2 f^{(4)\prime\prime}_\xi\right\| \\
    &\to \frac{7}{2} \cdot \frac{5}{8} = \frac{35}{16}
\end{split}
\end{equation}
when $\xi \to 1$. From the definition of $\hat{Z}^{(3)}$, Eq.~\eqref{eq:ZZ3}, we have
\begin{equation}
    \inf \langle \hat{Z}^{(6)} \rangle = \frac{1}{560}\left(\frac{35}{16}\right)^2 = \frac{35}{4096},
\end{equation}
which concludes the proof.
\end{proof}

The functions $f^{(6)}_\xi(r)$ that minimize $\langle\hat{Z}^{(3)}\rangle$ can now be taken as
\begin{equation}
    f^{(6)}_\xi(r) = \tilde{h}_\xi(r).
\end{equation}
The corresponding function $\Psi(r)$ is shown in Fig.~\ref{fig:psi3} and the corresponding wave functions read as
\begin{equation}
    \psi(x_1, \ldots, x_6) = \Psi(r) = \sqrt{\frac{6}{\pi^3}} f^{(6)}_\xi((x^2_1 + \ldots + x^2_6)^3).
\end{equation}
As in the previous case,
we have not been able to find an explicit expression for $f^{(6)}_\xi(r)$ in terms of known special function, but we could
prove that these functions have the desired property without such an expression.

\begin{figure}[h]
\includegraphics[scale=0.8]{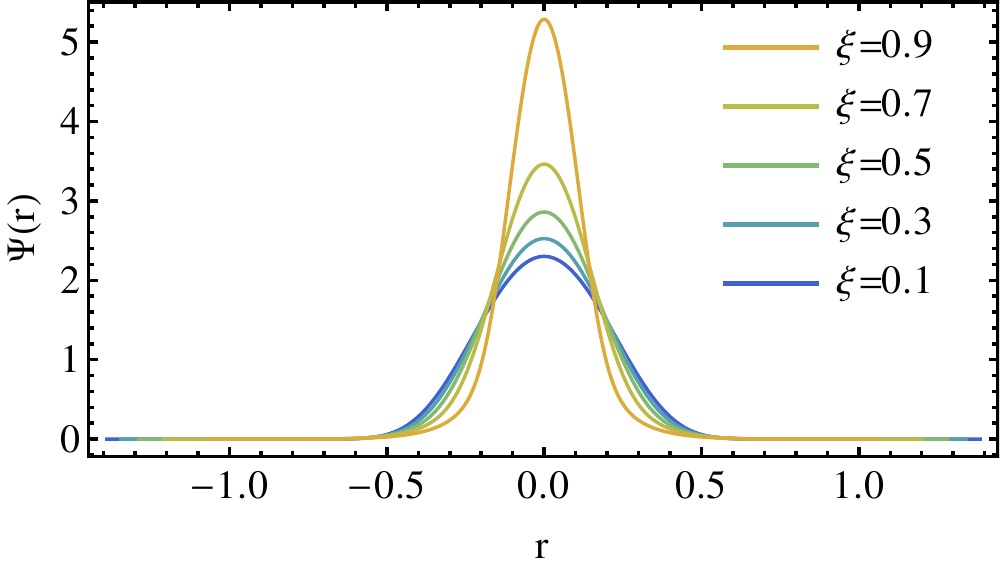}
\caption{The wave function $\Psi(r)$, defined by Eqs.~\eqref{eq:fn2} for $n=3$, \eqref{eq:htilde}, \eqref{eq:hxi} and \eqref{eq:Hxi} 
for a few different values of $\xi$.}\label{fig:psi3}
\end{figure}

\section{Conclusion}

In conclusion, we summarize our results obtained in this work. We have tried to minimize the uncertainty product $\sigma_{xp} \sigma_{px}$
and its multipartite generalizations. We have shown that pure states are sufficient to minimize such uncertainty products. We have 
outlined the general approach for the case of an arbitrary even number of parties and explicitly found
the infimum of these products over a special class of states with real spherically symmetric wave functions in bipartite, four-partite and six-partite
cases. We have also constructed parametrized families of states that approach the corresponding infimums by varying the parameter. Namely, we have 
shown that any bipartite state for which
\begin{equation}
    \frac{1}{8} < \sigma_{xp} \sigma_{px} < \frac{1}{4}
\end{equation}
is entangled; that any four-partite state for which 
\begin{equation}
    \frac{1}{30} < \sigma_{xxpp} \sigma_{ppxx} < \frac{1}{16}
\end{equation}
is entangled; and that any six-partite state such that
\begin{equation}
    \frac{35}{64} \cdot \frac{1}{64} = \frac{35}{4096} < \sigma_{xxxppp} \sigma_{pppxxx} < \frac{1}{64}
\end{equation}
is entangled and provided examples of states that tend to the limits of $1/8$, $1/30$ and $35/4096$ respectively. 
The next step will be to answer the question whether these limiting values can be improved by using more general states.

\appendix

\section{Changing the order of integration}

We need to verify that we can exchange summation and integration in Eq.~\eqref{eq:ff}. We use the following theorem that
can be found, for example, in \cite{budak}.
\begin{thrm}
Assume that the sequence of functions $F_n(t)$, $n = 0, 1, \ldots$, defined on an infinite interval $[t_0, +\infty)$, uniformly converges 
to the function $F(t)$ on any finite interval $[t_0, T]$, $T > t_0$. If the integral $\int^{+\infty}_{t_0} F_n(t)\,dt$ converges uniformly
with respect to $n$, then the integral $\int^{+\infty}_{t_0} F(t)\,dt$ exists and the equality 
\begin{equation}
    \lim_{n \to +\infty} \int^{+\infty}_{t_0} F_n(t)\,dt = \int^{+\infty}_{t_0} F(t)\,dt.
\end{equation}
is valid. In other words, under these conditions the limit of integrals is equal to the integral of the point-wise limiting function.
\end{thrm}
We apply this theorem to the functions
\begin{equation}
\begin{split}
    F_n(t) &= c_0 e^r S_n(\xi, t)J_0(2\sqrt{rt})e^{-t} \\
           &= c_0 e^r \sqrt{t}e^{-t}S_n(\xi, t) \frac{J_0(2\sqrt{rt})}{\sqrt{t}},
\end{split}
\end{equation}
where $\xi$ and $r$ are fixed. From the representation \eqref{eq:S} it is clear that the partial sum $S_n(\xi, t)$
converges uniformly to $I_0(t\sqrt{\xi})$ on any finite interval. To prove that the integral $\int^{+\infty}_0 F_n(t) \, dt$
converges uniformly note that 
\begin{equation}
    0 < \sqrt{t} e^{-t} S_n(\xi, t) < \sqrt{t} e^{-t} I_0(t\sqrt{\xi}).
\end{equation}
The Bessel function $I_0(t)$ has the following asymptotic expansion for large $t$ \cite{gr-I0}:
\begin{equation}
    I_0(t) \sim \frac{e^t}{\sqrt{2\pi t}} 
    \left(1 + \frac{1}{8t} + \frac{9}{128t^2} + \ldots\right),
\end{equation}
from which we immediately get that $\sqrt{t} e^{-t} I_0(t\sqrt{\xi}) \to 0$ when $t \to +\infty$ (take into account that $0 < \xi < 1$)
and, as a consequence, we see that the function $\sqrt{t} e^{-t} I_0(t\sqrt{\xi})$ is bounded on $[0, +\infty)$. We have just proved that
the functions $\sqrt{t} e^{-t} S_n(\xi, t)$ are uniformly bounded. On the other hand, the function $J_0(2\sqrt{rt})/\sqrt{t}$ 
(which does not depend on $n$) is integrable for $r>0$ \cite{gr-J0-2}
\begin{equation}
    \int^{+\infty}_0 \frac{J_0(2\sqrt{rt})}{\sqrt{t}} \, dt = 2\int^{+\infty}_0 J_0(2t\sqrt{r}) \, dt = \frac{1}{\sqrt{r}}.
\end{equation}
From this we can conclude that the integral $\int^{+\infty}_0 F_n(t) \, dt$ converges uniformly and thus we can apply the Theorem~1
to verify the correctness of the expression \eqref{eq:fr}.

\section{Normalization of the coefficients $c_{nm}$}\label{app:cnm}

Here we prove that the normalization condition for the coefficients defined by Eq.~\eqref{eq:cnm},
\begin{equation}\label{eq:c1}
    \sum^{+\infty}_{n,m=0} c^2_{nm}=1 ,
\end{equation}
can be obtained directly from the explicit expression \eqref{eq:cnm2} for these coefficients.
In fact, we have
\begin{equation}
    \sum^{+\infty}_{n,m=0} c^2_{nm} = \sum^{+\infty}_{N = 0} \sum_{n+m=N} c^2_{nm} = \sum^{+\infty}_{N = 0} \sideset{}{'}\sum_{n+m=4N} c^2_{nm},
\end{equation}
where the prime means that the sum runs over only those pairs of $n$ and $m$ that have the form $n = 4n'$, $m = 4m'$ or $n = 4n'+2$, $m = 4m'+2$. 
The first equality is simply the "diagonal" summation of this double series with non-negative coefficients (so the sum does not depend on the
order of summation) and the second equality is valid since if $N$ is not a multiple of 4 then $c_{nm} = 0$ for all pairs $(n,m)$ with $n+m=N$. 
From the expression \eqref{eq:cnm2} for the coefficients we see that we need to compute the following primed sum:
\begin{equation}\label{eq:nm}
\begin{split}
    S_N = \sideset{}{'}\sum_{n+m=4N} &\binom{n}{\frac{n}{2}}\binom{m}{\frac{m}{2}} = \sum^{N}_{n=0} \binom{4n}{2n}\binom{4N-4n}{2N-2n} \\
    &+ \sum^{N-1}_{n=0} \binom{4n+2}{2n+1}\binom{4N-4n-2}{2N-2n-1}.
\end{split}
\end{equation}
To compute this sum, let us introduce the function $F(x)$ via the equality \cite{gr-F} 
\begin{equation}
    F(x) \equiv \sum^{+\infty}_{n=0} \binom{2n}{n} x^n = \frac{1}{\sqrt{1-4x}}.
\end{equation}
Then we can observe that the sum \eqref{eq:nm} can be written in a compact form with the help of $F(x)$ as 
\begin{equation}
    S_N = \sum^{2N}_{n=0} [x^n] F(x) [x^{2N-n}] F(x) = [x^{2N}] F^2(x),
\end{equation}
where $[x^n]F(x)$ is the coefficient of $x^n$ in the Taylor expansion of $F(x)$.
Thus, we have
\begin{equation}
    S_N = [x^{2N}] \frac{1}{1-4x} = 2^{4N}.
\end{equation}
Now we can finish the computation of the sum on the left-hand side of Eq.~\eqref{eq:c1}:
\begin{equation}
\begin{split}
    \sum^{+\infty}_{n,m=0} c^2_{nm} &= \frac{\pi}{2K(\xi)} \sum^{+\infty}_{N = 0} \binom{2N}{N}^2 S_N \left(\frac{\xi^2}{16^2}\right)^N \\
    &= \frac{\pi}{2K(\xi)} \sum^{+\infty}_{N = 0} \binom{2N}{N}^2 \left(\frac{\xi^2}{16}\right)^N = 1,
\end{split}
\end{equation}
where we have taken into account the relation \eqref{eq:Kxi}. This completes the proof of the equality \eqref{eq:c1}.

\section{Roots of Eq.~\eqref{eq:alpha-eq}}\label{app:roots}

Here we prove that the numbers \eqref{eq:alpha} are the roots of the equation \eqref{eq:alpha-eq} where the coefficients $b_k$ are given by Eq.~\eqref{eq:b2}.
We have
\begin{equation}\label{eq:alpha2}
\begin{split}
    &\sum^n_{k=1} b_k (\alpha)_k = \sum^n_{k=1} \sum^k_{j=1} \frac{(-1)^{k+j}}{k!} \binom{k}{j} \binom{jn}{n} (\alpha)_k \\
    &= \sum^n_{j=1} (-1)^j \binom{jn}{n} \sum^n_{k=j} \frac{(-1)^k}{k!} \binom{k}{j} (\alpha)_k,
\end{split}
\end{equation}
where we have exchanged the summation order according to the equality $\sum^n_{k=1} \sum^k_{j=1} = \sum^n_{j=1} \sum^n_{k=j}$. 
The inner sum on the right-hand side of Eq.~\eqref{eq:alpha2}
is easier to compute when $\alpha$ is a sufficiently large integer number. In this case we can write $(\alpha)_k = \alpha!/(\alpha-k)!$ and
transform this sum as follows:
\begin{equation}\label{eq:ain}
\begin{split}
    &\sum^n_{k=j} \frac{(-1)^k}{k!} \binom{k}{j} (\alpha)_k = \binom{\alpha}{j} \sum^n_{k=j} (-1)^k \binom{\alpha-j}{k-j} = \\
    &(-1)^n \binom{\alpha}{j} \binom{\alpha-j-1}{n-j} = \frac{(-1)^{n+1}}{n!}\frac{(\alpha)_{n+1}}{j-\alpha}\binom{n}{j},
\end{split}
\end{equation}
where we have used a simple relation for the binomial coefficients \cite{gr-bin2}. Note that
\begin{equation}
    \frac{(\alpha)_{n+1}}{j-\alpha} = \frac{\alpha(\alpha-1)\ldots(\alpha-n)}{j-\alpha}
\end{equation}
is a polynomial of $\alpha$ for $j = 1, \ldots, n$. Since the polynomial on the left-hand side of Eq.~\eqref{eq:ain} equals 
the polynomial on the right-hand side for all sufficiently large integer $\alpha$, these two polynomials must be same, and the equality
\eqref{eq:ain} must be valid for all $\alpha$. When we substitute the expression on the right-hand side of Eq.~\eqref{eq:ain} into Eq.~\eqref{eq:alpha2}
we see that we need to prove that
\begin{equation}
    S(\alpha) \equiv (\alpha)_{n+1} \sum^n_{j=1} \frac{(-1)^j}{j-\alpha} \binom{n}{j} \binom{jn}{n} = 0
\end{equation}
for $\alpha$ given by Eq.~\eqref{eq:alpha}. For $\alpha = 0$ this is clearly true since $(0)_{n+1} = 0$. For $\alpha = i/n$, $i = 1, \ldots, n-1$, we have
\begin{equation}\label{eq:Sin}
    S\left(\frac{i}{n}\right) = \frac{n}{n!}\left(\frac{i}{n}\right)_{n+1}\sum^n_{j=1} (-1)^j \binom{n}{j} \prod^{n-1}_{\substack{l=0 \\ l \not= i}} (nj-l),
\end{equation}
since the binomial coefficient $\binom{jn}{n}$ can be expanded as
\begin{equation}
    \binom{jn}{n} = \frac{1}{n!} \prod^{n-1}_{l=0} (jn-l)
\end{equation}
and $1/(j-\frac{i}{n}) = n/(nj-i)$ cancels one term of this product. Since $i \not= 0$ the term with $l=0$, i.e. $jn$, is always present. Then we can 
expand the product on the right-hand side of Eq.~\eqref{eq:Sin} and get
\begin{equation}\label{eq:Sin2}
    S\left(\frac{i}{n}\right) = \sum^{n-1}_{p = 1} A_p \sum^n_{j=1} (-1)^j \binom{n}{j} j^p = 0,
\end{equation}
where $A_p$ are some numbers and each inner sum is zero due to another standard relation for binomial coefficients \cite{gr-bin3}, 
\begin{equation}\label{eq:bi}
    \sum^n_{j=1} (-1)^j \binom{n}{j} j^p = \sum^n_{j=0} (-1)^j \binom{n}{j} j^p = 0
\end{equation}
for $p = 1, \ldots, n-1$. It is important to note here that if $i \not= 0$ then $p$ in the sum \eqref{eq:Sin2} starts from $1$ and in this case
the first equality in Eq.~\eqref{eq:bi} holds true.

\end{document}